%% file: main.tex
\documentclass[11pt]{article}
\usepackage{fullpage}
\usepackage{authblk}
\input{preamble.tex}
\begin{document}
\title{Exponents for Shared Randomness-Assisted Channel Simulation}

\author[1]{Aadil Oufkir}
\author[1]{Michael X. Cao}
\author[2,3,4]{Hao-Chung Cheng}
\author[1]{Mario Berta}
\affil[1]{\small{Institute for Quantum Information, RWTH Aachen University, Aachen, Germany}}
\affil[2]{\small{Department of Electrical Engineering, National Taiwan University, Taipei, Taiwan}}
\affil[3]{\small{National Center for Theoretical Sciences, Taipei, Taiwan}}
\affil[4]{\small{Hon Hai (Foxconn) Quantum Computing Center, New Taipei City, Taiwan}}

\maketitle
\begin{abstract}
We determine the exact error and strong converse exponents of shared randomness-assisted channel simulation in worst case total-variation distance.
Namely, we find that these exponents can be written as simple optimizations over the Rényi channel mutual information.
Strikingly, and in stark contrast to channel coding, there are no critical rates, allowing a tight characterization for arbitrary rates below and above the simulation capacity.
We derive our results by asymptotically expanding the meta-converse for channel simulation [Cao {\it et al.}, IEEE Trans.~Inf.~Theory (2024)], which corresponds to non-signaling assisted codes. We prove this to be asymptotically tight by employing the approximation algorithms from [Berta {\it et al.}, Proc.~IEEE ISIT (2024)], which show how to round any non-signaling assisted strategy to a strategy that only uses shared randomness. Notably, this implies that any additional quantum entanglement-assistance does not change the error or the strong converse exponents.
\end{abstract}
\section{Introduction}
Channel simulation is a fundamental task in information theory,  where the goal is to replicate the behavior of a noisy communication channel using alternative resources such as noiseless communication and shared randomness.
This task is pivotal in understanding the limits of communication systems and is one of the two extreme cases in the broader problem of channel interconversion (see, e.g.,\cite{sudan2019communication}), with the other being Shannon’s channel coding problem \cite{shannon1948mathematical}.
While the reverse Shannon theorem \cite{bennett2002entanglement, bennett2014quantum} characterized the asymptotic rate of communication cost required to simulate a noisy channel, recent research has shifted toward analyses in the finite blocklength regime, focusing on various deviation regimes~\cite{fang2019quantum, ramakrishnan2023moderate}.
In particular, asymptotic communication rates in the small- and moderate-deviation regimes have been investigated in~\cite{cao2024channel}.

In this paper, we extend the study of randomness-assisted channel simulation to the large deviation regime, where the communication rate $r$ deviates from the channel capacity by some constant.
When the rate $r$ exceeds the channel capacity, the distortion in the total variation distance (TVD) between the simulated and the target channels diminishes exponentially fast as the number of channels grows. 
Conversely, when $r$ is below capacity, the distortion tends to $1$ exponentially. 
In both cases, the rate of convergence can be characterized asymptotically by the Rényi information of the channel.
Specifically, we derive the error exponent and the strong converse exponent as follows (see Theorems~\ref{thm:NS:EE},~\ref{thm:NS:SCE},~\ref{thm:SR:EE} and~\ref{thm:SR:SCE}):
\begin{align}
    \nonumber
    \lim_{n\to \infty} -\frac{1}{n} \log\epsilon^\fnc{SR}(\floor{e^{nr}},W_{\rv{Y}|\rv{X}}^{\tensor n})
    &= \lim_{n\to \infty} -\frac{1}{n} \log\epsilon^\fnc{NS}(\floor{e^{nr}},W_{\rv{Y}|\rv{X}}^{\tensor n}) \\
    \label{eq:SR:EE}
    &= \sup_{\alpha\geq 0} \; \alpha\cdot\left( r - I_{\alpha+1}(W_{\rv{Y}|\rv{X}}) \right)\\
    \nonumber
    \lim_{n\to \infty} -\frac{1}{n} \log\left(1-\epsilon^\fnc{SR}(\floor{e^{nr}},W_{\rv{Y}|\rv{X}}^{\tensor n})\right)
    &= \lim_{n\to \infty} -\frac{1}{n} \log\left(1-\epsilon^\fnc{NS}(\floor{e^{nr}},W_{\rv{Y}|\rv{X}}^{\tensor n})\right) \\
    \label{eq:SR:SCE}
    &= \sup_{\alpha\in[0,1]} \; (1-\alpha)\cdot\left( -r + I_\alpha(W_{\rv{Y}|\rv{X}}) \right),
\end{align}
where $\epsilon^\fnc{SR}(\floor{e^{nr}},W_{\rv{Y}|\rv{X}}^{\tensor n})$ and $\epsilon^\fnc{NS}(\floor{e^{nr}},W_{\rv{Y}|\rv{X}}^{\tensor n})$ denote the minimal achievable distortions in TVD for simulating $W_{\rv{Y}|\rv{X}}^{\tensor n}$ using a noiseless channel with alphabet size at most $\floor{e^{nr}}$ with shared-randomness assistance and non-signaling assistance, respectively; and $I_\alpha(W_{\rv{Y}|\rv{X}})$ is the Rényi channel mutual information of order $\alpha$, \ie,
\begin{equation}
    I_\alpha(W_{\rv{Y}|\rv{X}}) \defeq \adjustlimits\sup_{p_\rv{X}\in\set{P}(\set{X})} \inf_{q_\rv{Y}\in\set{P}(\set{Y})} D_{\alpha}\infdiv*{p_\rv{X}\cdot W_{\rv{Y}|\rv{X}}}{p_\rv{X}\times q_\rv{Y}}
\end{equation}
for all $\alpha\in [0,\infty)$.
Here, $D_\alpha$ is the Rényi divergence of order $\alpha$, defined as
\begin{equation}\label{eq:def:renyi}
    D_\alpha\infdiv*{p}{q} \defeq \frac{1}{\alpha-1} \log\left(\sum_{x\in\set{X}} p(x)^{\alpha}q(x)^{1-\alpha}\right)
    \qquad \forall \alpha \ge 0.
\end{equation}
Our results for randomness-assisted channel simulation are derived indirectly by first analyzing the exponents for channel simulation assisted by a stronger resource of the non-signaling correlation, then connecting the non-signaling and randomness-assisted scenarios via rounding techniques~\cite{berta2024optimality}.
We commence the analysis by establishing an elegant one-shot characterization for the non-signaling simulation error (Proposition~\ref{prop:epsilon:NS:exact}).
Then, the lower bound to the non-signaling error exponent is derived via Markov's inequality, while the upper bound relies on the de Finetti reduction and the method of types.
The lower bound to the non-signaling strong converse exponent is based on the Chernoff bound, while the upper bound relies on the method of types and certain continuity arguments.
We note that in contrast to channel coding, the proofs of the converse of the error exponent and the achievability of the strong converse exponent for channel simulation are significantly more involved than the other directions.

Interestingly, the error exponent and the strong converse exponent remain unchanged whether the shared resources between the sender and receiver consist of classical randomness or a potent non-signaling resource.
Consequently, this invariance applies to any shared resources that exist ``between'' SR and NS, namely any shared mechanisms capable of producing shared randomness and that can be generated by non-signaling resources.
A specific example of this is unlimited shared entanglement in quantum theory.

The remainder of this paper is organized as follows. 
In Section~\ref{sec:related}, we discuss related works to channel simulation.
In Section~\ref{sec:exact}, we introduce a rigorous formulation for channel simulation and derive an exact expression for the minimal distortion in TVD for non-signaling channel simulation with a fixed message size constraint based on previous work~\cite{cao2024channel}.
In Section~\ref{sec:EE}, we study the error exponent for non-signaling channel simulation when the rate is above the capacity.
In Section~\ref{sec:SCE}, we study the strong converse exponent for non-signaling channel simulation when the rate is below the capacity.
In Section~\ref{sec:SR}, we explore the connection between the exponents of channel simulation under the non-signaling scenario and the shared randomness-assisted scenario.
Section~\ref{sec:conclusion} concludes the paper.


\subsection{Related Works}\label{sec:related}
Channel simulation is related to, and sometimes confused with, a range of tasks often collectively known as the ``distributed generation of correlations'', with prominent examples including Wyner's common information~\cite{wyner1975common}, strong coordination~\cite{cuff2010coordination,cuff2013distributed}, and channel synthesis~\cite{winter2002compression,harsha2010communication}.
These tasks focus on simulating a discrete memoryless channel with a \emph{fixed} independent and identically distributed (i.i.d.) input distribution, which then amounts to generating a specific input-output joint distribution.
Unlike the previous works, the present paper considers channel simulation agnostic to \emph{arbitrary} (possibly non-i.i.d.) input distributions and aims to recreate the conditional probability (\ie, the channel itself).

The study of the error exponent of some of the aforementioned tasks is available in the literature.
Particular examples include the study of Wyner's common information and channel synthesis with a fixed input, which boil down to approximating a specific channel output as a target.
A key tool to that end is via employing a random or deterministic codebook of limited size at channel input and mixing the corresponding channel outputs.
Those procedures are conventionally termed \emph{soft covering lemma}~\cite{wyner1975common} and \emph{channel resolvability}~\cite{han1993approximation, hayashi2006general}, respectively.
The exact error exponents for soft covering under TVD and relative entropy were derived in~\cite[Theorems 1]{yagli2019exact} and~\cite[Theorem 4]{parizi2016exact}, respectively.
We remark that those quantities are expressed in terms of the R\'enyi divergence of orders between $1$ and $2$.
However, the established error exponent for randomness-assisted channel simulation, \ie,~\eqref{eq:SR:EE}, is for all orders greater than $1$.
Hence, the error exponent for soft covering alone might not be sufficient to imply our result in~\eqref{eq:SR:EE}.

Recently, Ref.~\cite{li2021reliable} proved the error exponent for simulating quantum channels with free entanglement under the so-called purified distance for the low rate regime (\ie, rates above channel capacity and below the critical rate).
In the special case where the underlying quantum channel is classical, the above result does not imply \eqref{eq:SR:EE} because purified distance is not equal to TVD even in the commuting setting. Using Fuchs–van de Graaf inequality \cite{fuchs1999cryptographic}, we can deduce the converse error exponent of Theorem \ref{thm:SR:EE} from the converse error exponent in purified distance \cite{li2021reliable}. However, our converse error exponent is non-asymptotic with polynomial prefactor (see Proposition \ref{prop:EE-converse}) whereas the converse error exponent of \cite{li2021reliable} is asymptotic. 
Ref.~\cite{Yao24} improves upon Ref.~\cite{li2021reliable} and establishes the error exponent and strong converse exponent for simulating classical-quantum channels under the purified distance for all rates. The strong converse exponent of \cite{Yao24} is of different nature, \ie,  it involves an optimization over $\alpha\in [\frac{1}{2},1]$ and a prefactor $\frac{1-\alpha}{\alpha}$ instead of $\alpha\in [0,1]$ and $(1-\alpha)$, respectively. 

\subsection{Notations}

We use the following conventions and notation throughout this paper.
\begin{itemize}
    \item The logarithm (denoted $\log$) in this paper is nature-based and information is measured in nats.
    \item The set of integers between $1$ and an $M$ is denoted $[M]$, \ie, $[M]\defeq\{1,\ldots, M\}$. 
    \item Sets are denoted by calligraphic fonts, \eg, $\set{X}$ reads ``set $\set{X}$''.
    \item Random variables are denoted in sans serif fonts, \eg, $\rv{X}$ reads ``random variable $\rv{X}$''.
    \item Vectors are denoted in boldface letters \eg, $\mathbf{x}$, and $\rvs{X}$.
    In particular, we use the subscript and superscript to denote the starting and ending indexes of a vector.
    Namely, $\mathbf{x}_1^n\equiv (x_1, \ldots, x_n)$, and $\rvs{X}_1^n\equiv (\rv{X}_1,\ldots, \rv{X}_n)$.
    \item Given a discrete random variable $\rv{X}$, we denote its probability mass function (pmf) by $p_\rv{X}$.
    \item Given a finite set $\set{X}$, $\set{P}(\set{X})$ denotes the set of all pmfs on $\set{X}$.
    Given another finite set $\set{Y}$, $\set{P}(\set{Y}|\set{X})$ denotes the set of all conditional pmfs on $\set{Y}$ conditioned on $\set{X}$.
    \item For any joint pmf $p_{\rvs{A}_1^n}\in\set{P}(\set{A}^n)$ of $n$ random variables, each with the same alphabet, $p_{\rvs{A}_1^n}$ is said to be permutation-invariant if and only if 
    \begin{equation}
        p_{\rvs{A}_1^n}(a_1,\ldots, a_n) = p_{\rvs{A}_1^n}(a_{\pi_1}, \ldots, a_{\pi_n})
    \end{equation}
    for \emph{all} permutations $\mathbf{\pi}$ of $\{1,\ldots, n\}$.
\end{itemize}


\section{Channel Simulation and the Exact Distortion in TVD}\label{sec:exact}
In this section, we formally define the tasks of channel simulation and derive a one-shot expression of the minimal distortion in TVD for non-signaling channel simulation with constrained message size.
Given a channel $W_{\rv{Y}|\rv{X}}\in\set{P}(\set{Y}|\set{X})$, where $\set{X}$ and $\set{Y}$ are some finite sets, the goal of channel simulation is to construct another channel $\tilde{W}_{\rv{Y}|\rv{X}}\in\set{P}(\set{Y}|\set{X})$ \emph{approximating} $W_{\rv{Y}|\rv{X}}$ using a noiseless channel (\ie, an identity channel) with a \emph{small} alphabet size.
Naturally, one prefers a higher quality of the approximation.
More specifically, in this paper, we measure the distortion of the simulation in TVD and aim to minimize the following quantity as much as possible:
\begin{align}
\epsilon = \norm*{W_{\rv{Y}|\rv{X}}-\tilde{W}_{\rv{Y}|\rv{X}}}_\fnc{TVD} &\defeq \max_{x\in\set{X}}\; \norm*{W_{\rv{Y}|\rv{X}}(\cdot|x)-\tilde{W}_{\rv{Y}|\rv{X}}(\cdot|x)}_\fnc{TVD}
\end{align}
where the TVD between two pmf $p$ and $q$ on $\set{X}$ is defined as
\begin{align}
\norm*{p-q}_\fnc{TVD} &\defeq \frac{1}{2}\sum_{x\in\set{X}} \abs*{p(x)-q(x)}.
\end{align}
On the other hand, one would also like to use as little resources as possible for the above task.
Specifically, we would like to minimize or restrict the size of the alphabet $M$ of the identity channel $\id_M(\cdot|\cdot)\defeq \delta_{\cdot|\cdot}\in\set{P}([M]|[M])$ used in the construction of $\tilde{W}_{\rv{Y}|\rv{X}}$. 
Depending on the additional resources (other than the aforementioned $\id_M$), $\tilde{W}_{\rv{Y}|\rv{X}}$ is constructed differently as follows (see Figure~\ref{fig:simulation:task}).
\begin{itemize}
    \item In the non-assisted scenario, we have the simple construction as $\tilde{W}_{\rv{Y}|\rv{X}} \defeq \mathcal{D}\circ\id_M\circ\mathcal{E}$ where $\mathcal{E}\in\set{P}([M]|\set{X})$, $\mathcal{D}\in\set{P}(\set{Y}|[M])$ are some uncorrelated encoder and decoder, respectively.
    \item In the shared randomness-assisted scenario, we allow the encoder and the decoder to be coordinated by some shared random variable, \ie, 
    \begin{equation}
        \tilde{W}_{\rv{Y}|\rv{X}}(y|x) \defeq \sum_{s\in\set{S}} p_\rv{S}(s) \cdot \sum_{i,j\in[M]} \mathcal{D}(y|j,s) \cdot \id_M(j|i) \cdot \mathcal{E}(i|x,s)
    \end{equation}
    where $\mathcal{E}\in\set{P}([M]|\set{X}\times\set{S})$, $\mathcal{D}\in\set{P}(\set{Y}|[M]\times\set{S})$.
    Despite the fact that the shared random variable $\rv{S}$ is also a type of communication resource, we do not constrain $\rv{S}$ (\ie, $p_\rv{S}$ nor $\set{S}$) in this paper.
    Hence, such $\rv{S}$ is often referred to as unconstrained shared randomness.
    \item In the non-signaling-assisted scenario, we further allow the encoder and the decoder to take a form as a joint encoding-decoding map $\mathcal{N}_{\rv{IY}|\rv{XJ}}\in\set{P}([M]\times\set{Y}|\set{X}\times[M])$, and construct $\tilde{W}_{\rv{Y}|\rv{X}}$ as
    \begin{equation}
    \tilde{W}_{\rv{Y}|\rv{X}}(y|x) = \sum_{i,j\in[M]} \mathcal{N}_{\rv{IY}|\rv{XJ}}(i,y|x,j) \cdot \id_M(j|i),
    \end{equation}
    where we only restrict $N$ as a non-signaling map (see, \eg,~\cite{cubitt2011zero, matthews2012linear, fang2019quantum}), \ie
    \begin{align}
    \label{eq:ns:requirement:1}
    \sum_{y\in\set{Y}} N_{\rv{IY}|\rv{XJ}}(i,y|x,j) &= N_{\rv{I}|\rv{X}}(i|x) \qquad\forall j\in [M], \forall i \in [M], \forall x\in \set{X},\\
    \label{eq:ns:requirement:2}
    \sum_{i\in[M]} N_{\rv{IY}|\rv{XJ}}(i,y|x,j) &= N_{\rv{Y}|\rv{J}}(y|j) \qquad\forall x\in\set{X}, \forall y\in \set{Y}, \forall j\in [M].
    \end{align}
\end{itemize}
\begin{figure} 
\centering
  \subfloat[\label{fig:simulation:task:nonassisted}]{%
       \includegraphics[]{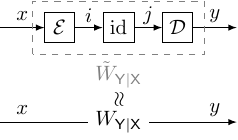}}\hfill
  \subfloat[\label{fig:simulation:task:random}]{%
        \includegraphics[]{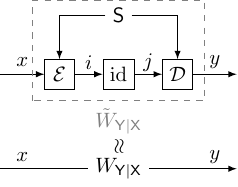}}\hfill
  \subfloat[\label{fig:simulation:task:NS}]{%
        \includegraphics[]{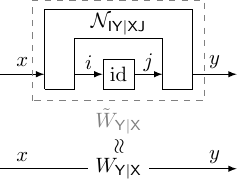}}
  \caption{Various channel simulation schemes: (a) Non-assisted Channel Simulation; (b) Randomness-assisted Channel Simulation; (c) Non-signaling-assisted Channel Simulation.
  (This figure is copied from~\cite{cao2024channel}.)}
  \label{fig:simulation:task} 
\end{figure}
Apparently, not all pairs $(\epsilon,M)\in[0,1]\times\mathbb{N}$ are feasible for the aforementioned tasks; and the study of channel simulation focuses on the trade-off between the accuracy of the simulation $\epsilon$ and its cost $M$.
In this paper, we focus on the minimal distortion $\epsilon$ for a given message size $M$.
For a given channel $W_{\rv{Y}|\rv{X}}\in\set{P}(\set{Y}|\set{X})$ and message size $M\in\mathbb{N}$, we denote the minimum attainable distortions in TVD in the above three scenarios by $\epsilon(M, W_{\rv{Y}|\rv{X}})$, $\epsilon^\fnc{SR}(M, W_{\rv{Y}|\rv{X}})$, and $\epsilon^\fnc{NS}(M, W_{\rv{Y}|\rv{X}})$, respectively.

In the remainder of this section, we focus on an explicit expression of $\epsilon^\fnc{NS}(M, W_{\rv{Y}|\rv{X}})$ as follows.
\begin{proposition}[Minimal distortion in TVD for non-signaling channel simulation] \label{prop:epsilon:NS:exact}
Given a channel $W_{\rv{Y}|\rv{X}}\in\set{P}(\set{Y}|\set{X})$ and a positive integer $M$, the minimal attainable distortion (measured in TVD) of non-signaling simulation codes for $W_{\rv{Y}|\rv{X}}$ with alphabet size at most $M$
 can be expressed as
\begin{align}
    \epsilon^\fnc{NS}(M,W_{\rv{Y}|\rv{X}})
    & = \adjustlimits \inf_{q_{\rv{Y}}\in\set{P}(\set{Y})} \max_{x\in\set{X}}\ 
    \sum_{y\in\set{Y}} \Big(W_{\rv{Y}|\rv{X}} (y|x) - M \cdot q_\rv{Y}(y)\Big)_+ \label{eq:epsilon:NS:exact:1} \\
    & = \adjustlimits \inf_{q_\rv{Y}\in\set{P}(\set{Y})} \sup_{p_\rv{X}\in\set{P}(\set{X})}\ 
    \sum_{x\in\set{X}, y\in\set{Y}}
    \Big( p_\rv{X}(x)\cdot W_{\rv{Y}|\rv{X}}(y|x) - M \cdot p_\rv{X}(x)\cdot q_\rv{Y}(y)\Big)_+ \label{eq:epsilon:NS:exact:2}
\end{align}
where $(x)_+\defeq \max\{x,0\}$ for $x\in \mathbb{R}$. 
\end{proposition}
We shall need the following lemma for the proof of the above proposition.
\begin{lemma}\label{lem:mim:tvd}
For any function $f:\set{A}\to\reals_{\geq 0}$ s.t. $\sum_{a\in\set{A}} f(a) \geq 1$ where $\set{A}$ is some finite set, it holds that
\begin{equation}
    \inf_{\tilde{p}\in\set{P}(\set{A}):\; \tilde{p} \leq f} \norm{\tilde{p} - p}_\fnc{TVD} = \sum_{a\in\set{A}} \Big( p(a) - f(a) \Big)_+
\end{equation}
for any pmf $p$ on $\set{A}$.
\end{lemma}
\begin{proof}
Firstly, LHS is not smaller, as LHS$=\inf_{\tilde{p}\in\set{P}(\set{A}):\; \tilde{p} \leq f} \sum_{a} (p(a)-\tilde{p}(a))_+\ge  \sum_{a} (p(a)-f(a))_+=$RHS.
To show that the LHS is not larger, we construct the following $\tilde{p}$:
\begin{equation}
    \tilde{p}(a) \defeq \begin{cases}
    f(a) & \text{for } a\in\set{A} \text{ s.t. } f(a)\leq p(a)\\
    p(a) + (f(a)-p(a))\cdot \frac{\sum_{a\in\set{A}}\big(p(a)-f(a)\big)_+}{\sum_{a\in\set{A}}\big(f(a)-p(a)\big)_+} &\text{otherwise}.
    \end{cases}
\end{equation}
It is straightforward to check that $\tilde{p}$ is a pmf on $\set{A}$,  $\tilde{p}\leq f$, and $\norm{\tilde{p}-p}_{\fnc{TVD}}=$RHS.
\end{proof}
\begin{proof}[Proof of Proposition~\ref{prop:epsilon:NS:exact}]
For any $\epsilon\in(0,1)$, let $\mathcal{M}^{\fnc{NS},\star}_{\epsilon}(W_{\rv{Y}|\rv{X}})$ denote the minimal attainable message size of non-signaling $\epsilon$-simulation (in TVD) codes for $W_{\rv{Y}|\rv{X}}$.
By definition, we can write
\begin{equation}
    \epsilon^\fnc{NS}(M,W_{\rv{Y}|\rv{X}}) = \inf\left\{\epsilon \geq 0 \middle\vert \mathcal{M}^{\fnc{NS},\star}_{\epsilon}(W_{\rv{Y}|\rv{X}}) \leq M \right\}.
\end{equation}
We cite Theorem~20 in~\cite{cao2024channel}, \ie,
\begin{equation}
\log{\mathcal{M}^{\fnc{NS},\star}_{\epsilon}(W_{\rv{Y}|\rv{X}})} = \ceil{I_{\max}^\epsilon(W_{\rv{Y}|\rv{X}})}_{\log{\mathbb{Z}_{>0}}},
\end{equation}
where 
\begin{align}
    I_{\max}^\epsilon(W_{\rv{Y}|\rv{X}})
    & \defeq \multiadjustlimits{
    \inf_{\tilde{W}_{\rv{Y}|\rv{X}}\in\set{P}(\set{Y}|\set{X}):\: \norm{\tilde{W}_{\rv{Y}|\rv{X}} - W_{\rv{Y}|\rv{X}}}_\fnc{TVD} \leq\epsilon},
    \sup_{p_\rv{X}\in\set{P}(\set{X})},
    \inf_{q_\rv{Y}\in\set{P}(\set{Y})},
    \max_{x\in\set{X}},
    \max_{y\in\set{Y}}}\ 
    \log{\frac{p_\rv{X}(x)\cdot\tilde{W}_{\rv{Y}|\rv{X}}(y|x)}{p_\rv{X}(x)\times q_\rv{Y}(y)}} \\
    & = \multiadjustlimits{
    \inf_{q_\rv{Y}\in\set{P}(\set{Y})}, 
    \inf_{\tilde{W}_{\rv{Y}|\rv{X}}\in\set{P}(\set{Y}|\set{X}):\: \norm{\tilde{W}_{\rv{Y}|\rv{X}} - W_{\rv{Y}|\rv{X}}}_\fnc{TVD} \leq\epsilon},
    \max_{x\in\set{X}},
    \max_{y\in\set{Y}}}\ 
    \log{\frac{\tilde{W}_{\rv{Y}|\rv{X}}(y|x)}{q_\rv{Y}(y)}}.
\end{align}
This enables us to write
\begin{equation}
    \epsilon^\fnc{NS}(M,W_{\rv{Y}|\rv{X}}) = \adjustlimits \inf_{q_\rv{Y}\in\set{P}(\set{Y})} \inf_{\tilde{W}_{\rv{Y}|\rv{X}}\in\set{P}(\set{Y}|\set{X}): \tilde{W}_{\rv{Y}|\rv{X}} \leq M \cdot q_{\rv{Y}}} \norm{\tilde{W}_{\rv{Y}|\rv{X}} - W_{\rv{Y}|\rv{X}}}_{\fnc{TVD}}
\end{equation}
for each positive integer $M$.
Using the definition that $\norm{\tilde{W}_{\rv{Y}|\rv{X}} - W_{\rv{Y}|\rv{X}}}_{\fnc{TVD}} = \max_{x\in\set{X}} \norm{\tilde{W}_{\rv{Y}|\rv{X}}(\cdot|x) - W_{\rv{Y}|\rv{X}}(\cdot|x)}_\fnc{TVD}$, and the fact that 
\begin{equation}
\min_{a_1\in\set{A}_1, a_2\in\set{A}_2, \ldots, a_k\in\set{A}_k}
\max_{j=1,\ldots, k} f_j(a_j) = \max_{j=1,\ldots,k} \min_{a\in\set{A}_j} f_j(a)
\end{equation}
for any positive integer $k$ and real-valued functions $f_1,\ldots, f_k$ that has a minimal point on each of the domains $\set{A}_1, \ldots, \set{A}_k$, respectively, we have
\begin{align}
    \epsilon^\fnc{NS}(M,W_{\rv{Y}|\rv{X}}) &=\multiadjustlimits{
    \inf_{q_\rv{Y}\in\set{P}(\set{Y})},
    \inf_{\tilde{W}_{\rv{Y}|\rv{X}}\in\set{P}(\set{Y}|\set{X}):\: \tilde{W}_{\rv{Y}|\rv{X}}\leq M \cdot q_{\rv{Y}}},
    \max_{x\in\set{X}}}\ 
    \norm{\tilde{W}_{\rv{Y}|\rv{X}}(\cdot|x) - W_{\rv{Y}|\rv{X}}(\cdot|x)}_{\fnc{TVD}}\\ 
    \label{eq:sawp:inf:min}
    &=\multiadjustlimits{
    \inf_{q_\rv{Y}\in\set{P}(\set{Y})},
    \min _{\tilde{W}_{\rv{Y}|\rv{X}}\in\set{P}(\set{Y}|\set{X}):\: \tilde{W}_{\rv{Y}|\rv{X}}\leq M \cdot q_{\rv{Y}}},
    \max_{x\in\set{X}}}\ 
    \norm{\tilde{W}_{\rv{Y}|\rv{X}}(\cdot|x) - W_{\rv{Y}|\rv{X}}(\cdot|x)}_{\fnc{TVD}}\\ 
    &= \multiadjustlimits{
    \inf_{q_\rv{Y}\in\set{P}(\set{Y})},
    \max_{x\in\set{X}},
    \min_{\tilde{W}_{\rv{Y}|\rv{X}}(\cdot|x)\in\set{P}(\set{Y}):\: \tilde{W}_{\rv{Y}|\rv{X}}(\cdot|x) \leq M \cdot q_{\rv{Y}}} }
    \norm{\tilde{W}_{\rv{Y}|\rv{X}}(\cdot|x) - W_{\rv{Y}|\rv{X}}(\cdot|x)}_{\fnc{TVD}}.
\end{align}
where in~\eqref{eq:sawp:inf:min} we replace an infimum by minimum as the inner function is continuous and the domain is compact.
Eq.~\eqref{eq:epsilon:NS:exact:1} can be obtained from the above expression using Lemma~\ref{lem:mim:tvd}.
To obtain~\eqref{eq:epsilon:NS:exact:2}, we relax the domain of the maximization and write 
\begin{equation}
    \epsilon^\fnc{NS}(M,W_{\rv{Y}|\rv{X}})
    \leq \adjustlimits \inf_{q_{\rv{Y}}\in\set{P}(\set{Y})} \sup_{p_{\rv{X}}\in\set{P}(\set{X})}\ 
    \underbrace{\sum_{x\in\set{X}, y\in\set{Y}}
    \Big(p_\rv{X}(x)\cdot W_{\rv{Y}|\rv{X}}(y|x) - M \cdot p_\rv{X}(x)\cdot q_\rv{Y}(y)\Big)_+}_{\text{linear in $p_\rv{X}$}}.
\end{equation}
Note that the target function is linear in $p_\rv{X}$, and thus the maximization over $p_\rv{X}$ is attained at extreme points, \ie,~\eqref{eq:epsilon:NS:exact:2} holds.
\end{proof}
\section{Error Exponent for Non-Signaling Channel Simulation}\label{sec:EE}
In the following, we are interested in the behavior of $\epsilon^\fnc{NS}$ for simulating asymptotically many copies of a channel given a fixed communication \emph{rate} $r>0$.
More precisely, we would like to study the exponent $-\frac{1}{n}\log{\epsilon^\fnc{NS}(\floor{e^{nr}},W_{\rv{Y}|\rv{X}}^{\tensor n})}$ as $n$ tends to infinity.
Note that we will omit the second argument for $\epsilon^\fnc{NS}$, \ie, $W_{\rv{Y}|\rv{X}}^{\tensor n}$ when it is clear from the context. 
The first result of this section, Proposition~\ref{prop:EE-Ach}, shows that $\epsilon^\fnc{NS}$ exponentially decays as $r > I_1(W_{\rv{Y}|\rv{X}})$ for any finite blocklength $n$.
Then, we will prove in Proposition~\ref{prop:EE-converse} that the achievable error exponent is asymptotically tight.

\subsection{Achievability for the Error Exponent}\label{sec:EE-Ach}
\begin{proposition}[Achievability for the Error Exponents]\label{prop:EE-Ach}
Let $W_{\rv{Y}|\rv{X}}\in\set{P}(\set{Y}|\set{X})$ be a channel.
For all $r> 0$, it holds that 
\begin{equation}\label{eq:EE-Ach-finite}
     \frac{1}{n} \log\epsilon^\fnc{NS}(\floor{e^{nr}},W_{\rv{Y}|\rv{X}}^{\tensor n}) \leq
    - \sup_{\alpha\geq 0} \; \alpha\cdot\left( r_n - \adjustlimits\sup_{p_\rv{X}\in\set{P}(\set{X})} \inf_{q_\rv{Y}\in\set{P}(\set{Y})} D_{\alpha+1}\infdiv*{p_\rv{X}\cdot W_{\rv{Y}|\rv{X}}}{p_\rv{X}\times q_\rv{Y}} \right)
\end{equation}
for all positive integer $n$, where $r_n\defeq \frac{1}{n}\log{\floor{e^{nr}}}$.
Moreover, as $n$ tends to infinity we have that
\begin{equation}\label{eq:EE-Ach}
    \limsup_{n\to \infty} \frac{1}{n} \log\epsilon^\fnc{NS}(\floor{e^{nr}},W_{\rv{Y}|\rv{X}}^{\tensor n}) \leq
    - \sup_{\alpha\geq 0} \; \alpha\cdot\left( r - \adjustlimits\sup_{p_\rv{X}\in\set{P}(\set{X})} \inf_{q_\rv{Y}\in\set{P}(\set{Y})} D_{\alpha+1}\infdiv*{p_\rv{X}\cdot W_{\rv{Y}|\rv{X}}}{p_\rv{X}\times q_\rv{Y}} \right).
\end{equation}
\end{proposition}
\begin{remark} \label{remark:EE}
The established achievable error exponent $\sup_{\alpha\geq 0} \alpha \left( r - I_{\alpha+1}(W_{\rv{Y}|\rv{X}}) \right)$ is expressed in terms of the R\'enyi information of channel, $I_\alpha(W_{\rv{Y}|\rv{X}})$.
Moreover, since $I_\alpha(W_{\rv{Y}|\rv{X}})$ is continuous in $\alpha \in [0,\infty]$ \cite[Lemma 16.(g)]{nakibouglu2018renyi}, the achievable error exponent is positive if and only if the communication rate is strictly greater than the mutual information of channel (which also coincides with Shannon's channel capacity), \ie,
$r > I_1(W_{\rv{Y}|\rv{X}})$.
Proposition~\ref{prop:EE-Ach} then indicates that the non-signaling simulation error $\epsilon^{\text{NS}}$ decays exponentially for any $r > I_1(W_{\rv{Y}|\rv{X}})$.
\end{remark} 
\begin{proof} 
By Proposition~\ref{prop:epsilon:NS:exact}, we have
\begin{equation}
    \epsilon^\fnc{NS}(\floor{e^{nr}},W_{\rv{Y}|\rv{X}}^{\tensor n}) = \adjustlimits \inf_{q_{\rvs{Y}_1^n}\in\set{P}(\set{Y}^n)} \max_{\mathbf{x}_1^n\in\set{X}^n}\ 
    \sum_{\mathbf{y}_1^n\in\set{Y}^n} \Big(W_{\rv{Y}|\rv{X}}^{\tensor n} (\mathbf{y}_1^n|\mathbf{x}_1^n) - \floor{e^{nr}} \cdot q_{\rvs{Y}_1^n}(\mathbf{y}_1^n)\Big)_+ .
\end{equation}
Restricting the first infimum to i.i.d. $q_{\rvs{Y}_1^n}$'s, we have
\begin{equation}
    \epsilon^\fnc{NS}(\floor{e^{nr}},W_{\rv{Y}|\rv{X}}^{\tensor n})\leq \adjustlimits \inf_{q_\rv{Y}\in\set{P}(\set{Y})} \max_{\mathbf{x}_1^n\in\set{X}^n}\ 
    \sum_{\mathbf{y}_1^n\in\set{Y}^n} \Big(W_{\rv{Y}|\rv{X}}^{\tensor n} (\mathbf{y}_1^n|\mathbf{x}_1^n) - \floor{e^{nr}} \cdot q_\rv{Y}^{\tensor n}(\mathbf{y}_1^n)\Big)_+ 
\end{equation}
We further have the following chain of equalities/inequalities for $\alpha\ge 0$
\begin{align}
    \text{Above} &= \adjustlimits \inf_{q_\rv{Y}\in\set{P}(\set{Y})} \max_{\mathbf{x}_1^n\in\set{X}^n}\ 
    \sum_{\mathbf{y}_1^n\in\set{Y}^n: W_{\rv{Y}|\rv{X}}^{\tensor n} (\mathbf{y}_1^n|\mathbf{x}_1^n) \geq \floor{e^{nr}} \cdot q_\rv{Y}^{\tensor n}(\mathbf{y}_1^n)} W_{\rv{Y}|\rv{X}}^{\tensor n} (\mathbf{y}_1^n|\mathbf{x}_1^n) - \floor{e^{nr}} \cdot q_\rv{Y}^{\tensor n}(\mathbf{y}_1^n)  \\
    &\leq \adjustlimits \inf_{q_\rv{Y}\in\set{P}(\set{Y})} \max_{\mathbf{x}_1^n\in\set{X}^n}\ 
    \sum_{\mathbf{y}_1^n\in\set{Y}^n: W_{\rv{Y}|\rv{X}}^{\tensor n} (\mathbf{y}_1^n|\mathbf{x}_1^n) \geq \floor{e^{nr}} \cdot q_\rv{Y}^{\tensor n}(\mathbf{y}_1^n)} W_{\rv{Y}|\rv{X}}^{\tensor n} (\mathbf{y}_1^n|\mathbf{x}_1^n) \\
    &= \adjustlimits \inf_{q_\rv{Y}\in\set{P}(\set{Y})} \max_{\mathbf{x}_1^n\in\set{X}^n}\ 
    \sum_{\mathbf{y}_1^n\in\set{Y}^n: W_{\rv{Y}|\rv{X}}^{\tensor n} (\mathbf{y}_1^n|\mathbf{x}_1^n) \geq \floor{e^{nr}} \cdot q_\rv{Y}^{\tensor n}(\mathbf{y}_1^n)} \left(W_{\rv{Y}|\rv{X}}^{\tensor n} (\mathbf{y}_1^n|\mathbf{x}_1^n)\right)^{\alpha+1} \!\cdot \left(W_{\rv{Y}|\rv{X}}^{\tensor n} (\mathbf{y}_1^n|\mathbf{x}_1^n)\right)^{-\alpha}\!\!\\
    &\leq \adjustlimits \inf_{q_\rv{Y}\in\set{P}(\set{Y})} \max_{\mathbf{x}_1^n\in\set{X}^n}\ 
    \sum_{\mathbf{y}_1^n\in\set{Y}^n: W_{\rv{Y}|\rv{X}}^{\tensor n} (\mathbf{y}_1^n|\mathbf{x}_1^n) \geq \floor{e^{nr}} \cdot q_\rv{Y}^{\tensor n}(\mathbf{y}_1^n)} \left(W_{\rv{Y}|\rv{X}}^{\tensor n} (\mathbf{y}_1^n|\mathbf{x}_1^n)\right)^{\alpha+1} \!\cdot \Big(\floor{e^{nr}} \!\cdot\! q_\rv{Y}^{\tensor n}(\mathbf{y}_1^n)\Big)^{-\alpha}\!\!\!\!\! \\
    &\leq \adjustlimits \inf_{q_\rv{Y}\in\set{P}(\set{Y})} \max_{\mathbf{x}_1^n\in\set{X}^n}\ 
    \sum_{\mathbf{y}_1^n\in\set{Y}^n} \left(W_{\rv{Y}|\rv{X}}^{\tensor n} (\mathbf{y}_1^n|\mathbf{x}_1^n)\right)^{\alpha+1} \cdot \Big(\floor{e^{nr}} \cdot q_\rv{Y}^{\tensor n}(\mathbf{y}_1^n)\Big)^{-\alpha}.
\end{align}
Taking logarithm and dividing by $n$ on both sides, we have for $\alpha\ge 0$
\begin{align}
    \label{eq:EE-Ach:C:1}
    \frac{1}{n} \log\epsilon^\fnc{NS}(\floor{e^{nr}},W_{\rv{Y}|\rv{X}}^{\tensor n}) 
    &\leq -\alpha\cdot\frac{\log{\floor{e^{nr}}}}{n}
    + \adjustlimits \inf_{q_\rv{Y}\in\set{P}(\set{Y})} \max_{\mathbf{x}_1^n\in\set{X}^n}\
    \frac{1}{n} \log{\prod_{i=1}^n \left(\sum_{y\in\set{Y}} \Big(W_{\rv{Y}|\rv{X}} (y|x_i)\Big)^{\alpha+1} \!\!\!\cdot \Big(q_\rv{Y}(y)\Big)^{-\alpha}\!\right)} \\ 
    \label{eq:EE-Ach:C:2}
    &= -\alpha\cdot\frac{\log{\floor{e^{nr}}}}{n}
    + \adjustlimits \inf_{q_\rv{Y}\in\set{P}(\set{Y})} \max_{\mathbf{x}_1^n\in\set{X}^n}\
    \frac{1}{n}\sum_{i=1}^n \alpha\cdot D_{\alpha+1}\infdiv*{W_{\rv{Y}|\rv{X}}(\cdot|x_i)}{q_\rv{Y}} \\ 
    \label{eq:EE-Ach:C:3}
    &= -\alpha\cdot\frac{\log{\floor{e^{nr}}}}{n}
    + \alpha\cdot \adjustlimits \inf_{q_\rv{Y}\in\set{P}(\set{Y})} \max_{x\in\set{X}}\
    D_{\alpha+1}\infdiv*{W_{\rv{Y}|\rv{X}}(\cdot|x)}{q_\rv{Y}} \\
    \label{eq:EE-Ach:C:4}
    &= -\alpha\cdot\frac{\log{\floor{e^{nr}}}}{n}
    + \alpha\cdot \adjustlimits \inf_{q_\rv{Y}\in\set{P}(\set{Y})} \sup_{p_\rv{X}\in\set{P}(\set{X})}\
    D_{\alpha+1}\infdiv*{p_\rv{X}\cdot W_{\rv{Y}|\rv{X}}}{p_\rv{X}\times q_\rv{Y}} 
\end{align}
where we distribute $y_1,\ldots, y_n$ into the products in~\eqref{eq:EE-Ach:C:1}; shuffle the maximization inside the summation in~\eqref{eq:EE-Ach:C:3}; and utilize the fact that the logarithm function is concave in~\eqref{eq:EE-Ach:C:4}\footnote{More precisely, $\max_{x\in\set{X}} \log{f(x)} = \sup_{p_\rv{X}\in\set{P}(\set{X})} \log{\sum_{x\in\set{X}}p_\rv{X}(x)\cdot f(x)}$ for any non-negative function $f$ on $\set{X}$.}. We deduce \eqref{eq:EE-Ach-finite} by taking the infimum over $\alpha\ge 0$ in \eqref{eq:EE-Ach:C:4}.

Taking $n\to\infty$ on both sides, we have for $\alpha\ge 0$
\begin{equation}
    \limsup_{n\to \infty} \frac{1}{n} \log\epsilon^\fnc{NS}(\floor{e^{nr}},W_{\rv{Y}|\rv{X}}^{\tensor n}) \leq
    -\alpha\cdot\left( r - \adjustlimits\sup_{p_\rv{X}\in\set{P}(\set{X})} \inf_{q_\rv{Y}\in\set{P}(\set{Y})} D_{\alpha+1}\infdiv*{p_\rv{X}\cdot W_{\rv{Y}|\rv{X}}}{p_\rv{X}\times q_\rv{Y}} \right).
\end{equation}
Finally,~\eqref{eq:EE-Ach} holds since the previous inequality  is true for all $\alpha\geq 0$.
\end{proof}
\subsection{Converse for the  Error Exponent}\label{sec:EE-converse}
As the other direction, we prove the following converse statement.
\begin{proposition}\label{prop:EE-converse}
Let $W_{\rv{Y}|\rv{X}}\in\set{P}(\set{Y}|\set{X})$ be a channel.
For all $r>0$, it holds that 
\begin{equation}
    \frac{1}{n}\log{\epsilon^\fnc{NS}(\floor{e^{nr}})} \geq
    -\sup_{\alpha\geq 0}\; \alpha\cdot \left( r+ g_n - 
    \adjustlimits\inf_{q_\rv{Y}\in\set{P}(\set{Y})} \sup_{p_\rv{X}\in\set{P}(\set{X})}
    D_{1+\alpha}\infdiv*{p_\rv{X}\cdot W_{\rv{Y}|\rv{X}}}{p_\rv{X}\times q_\rv{Y}}\right)
    - f_n
\end{equation}
for all integer $n\geq 3$, some non-negative sequences $\{f_n\}_n$ and $\{g_n\}_n$ such that $f_n=O\left(\frac{\log{n}}{n}\right)$ and $g_n=O\left(\frac{\log{n}}{n}\right)$.
In particular, as $n$ tends to infinity, we have
\begin{align}
    \liminf_{n\rightarrow \infty } \frac{1}{n}\log\epsilon^\fnc{NS}(\floor{e^{nr}})&\ge -\sup_{\alpha\geq 0}\; \alpha\cdot\left( r -
    \adjustlimits\inf_{q_\rv{Y}\in\set{P}(\set{Y})} \sup_{p_\rv{X}\in\set{P}(\set{X})}
    D_{1+\alpha}\infdiv*{p_\rv{X}\cdot W_{\rv{Y}|\rv{X}}}{p_\rv{X}\times q_\rv{Y}} \right).
\end{align}
\end{proposition}
We shall need the following post-selection lemma (\aka de Finetti reduction) to prove the above statement.
\begin{lemma}[Post-Selection Lemma / de Finetti reduction]\label{lem:perm-inv-de-finetti}
Let $\set{A}$ be a finite set.
There exists a probability measure $\nu$ on the set of all pmfs on $\set{A}$, \ie, $\set{P}(\set{A})$, such that for all positive integers $n$  
\begin{equation}
    p_{\rvs{A}_1^n}\leq \binom{n+\size*{\set{A}}-1}{n} \cdot \int \diff{\nu}(q_\rv{A}) q_\rv{A}^{\tensor n}
\end{equation}
for every permutation invariant pmf $p_{\rvs{A}_1^n}\in\set{P}(\set{A}^n)$.
\end{lemma}
We defer the proof of the above lemma to Appendix~\ref{app:proof:lem:perm-inv-de-finetti}.
Interested readers may want to refer to~\cite[Lemma 14]{tomamichel2017operational} for a similar result with a so-called ``universal probability distribution'' on the RHS.
The above integral form  can be seen as a classical counterpart of~\cite{christandl2009postselection, hayashi2010universal}.
As shown in Remark~\ref{app:rmk:lem:perm-inv-de-finetti}, the universal probability distribution as in~\cite{tomamichel2017operational} actually admits such an integral form as well.

We also utilize the method of types in proving the above proposition.
We shall need the following notation and the lemmas thereafter.
\begin{definition}[Types]\label{def:type}
    Given a finite set $\set{X}$ and a positive integer $n$, the \emph{set of types with denominator $n$ of the alphabet $\set{X}$} is the following subset of pmfs on $\set{X}$:
    \begin{equation}
        \set{P}_n(\set{X}) \defeq \left\{p\in\set{P}(\set{X})\;\middle\vert\; n\cdot p(x) \in\integers\ \forall x\in\set{X}\right\}.
    \end{equation}
    Furthermore, the type of a sequence $\mathbf{x}_1^n\in\set{X}^n$, denoted by $p_\rv{X}^{(\mathbf{x}_1^n)}$, is defined as the empirical distribution induced by the sequence, \ie,
    \begin{equation}
        p_\rv{X}^{(\mathbf{x}_1^n)} \defeq \frac{1}{n} \size*{\left\{i\in[n] \;\middle\vert\; x_i = x\right\}}.
    \end{equation}
    Finally, for a type $p\in\set{P}_n(\set{X})$, we denote the set of sequences of length $n$ and type $p$ by $\set{T}_n(p)$, \ie, 
    \begin{equation}
    \set{T}_n(p)\defeq\left\{ \mathbf{x}_1^n\in\set{X}^n \;\middle\vert\; p_\rv{X}^{(\mathbf{x}_1^n)} = p \right\}.
    \end{equation}
    The set $\set{T}_n(p)$ is also referred to as the type class of $p$.
    (Note that $p_\rv{X}^{(\mathbf{x}_1^n)}\in\set{P}(\set{X})$, and should be distinguished from $p_{\rvs{X}_1^n}$ where the latter is a distribution on $\set{X}^n$.)
\end{definition}

As a first observation, the set of types $\set{P}_n(\set{X})$ servers as an approximation of $\set{P}(\set{X})$ as stated in the following lemma.
\begin{lemma}[\cite{shannon1956zero}]\label{lem:type:approx}
Let $\set{X}$ be a finite set, and let $n$ be a positive integer.
For every $p_\rv{X}\in\set{P}(\set{X})$ there exists some $p_\rv{X}^{(n)}\in\set{P}_n(\set{X})$ such that $\abs*{p_\rv{X}(x)-p_\rv{X}^{(n)}(x)}\leq\frac{1}{n}$ for all $x\in\set{X}$.
\end{lemma}
\begin{proof}
This has been pointed out in~\cite{shannon1956zero}.
The idea is that, for any $p_\rv{X}\in\set{P}(\set{X})$,  there always exists a sequence $\mathbf{x}_1^n\in\set{X}_1^n$ such that
\begin{equation}
\floor{n\cdot p_\rv{X}} \leq n\cdot p_\rv{X}^{(\mathbf{x}_1^n)} \leq \ceil{n\cdot p_\rv{X}},
\end{equation}
and the type $p_\rv{X}^{(\mathbf{x}_1^n)}$ satisfies the lemma.
\end{proof}
\begin{definition}[Conditional Types]\label{def:conditonal:type}
    Given finite sets $\set{X}$, $\set{Y}$ and positive integer $n$, the \emph{set of conditional types with denominator $n$ of the alphabet $\set{Y}$ over $\set{X}$} is the following subset of conditional pmfs on $\set{Y}$ over $\set{Y}$:
    \begin{equation}
        \set{P}_n(\set{Y}|\set{X}) \defeq\left\{p_{\rv{Y}|\rv{X}}\in\set{P}(\set{Y}|\set{X}) \;\middle\vert\; p_\rv{X}\cdot p_{\rv{Y}|\rv{X}} \in\set{P}_n(\set{Y}\times\set{X})\ \exists p_\rv{X}\in\set{P}_n(\set{X})\right\}.
    \end{equation}
    Furthermore, the conditional type of sequence $\{(x_i,y_i)\}_{i=1}^n$ (more often specified as $(\mathbf{x}_1^n,\mathbf{y}_1^n)$), denoted by $p_{\rv{Y}|\rv{X}}^{(\mathbf{x}_1^n,\mathbf{y}_1^n)}$, is defined as the empirical conditional distribution induced by the sequences, \ie,
    \begin{equation}
        p_{\rv{Y}|\rv{X}}^{(\mathbf{x}_1^n,\mathbf{y}_1^n)}(y|x) \defeq \frac{p_\rv{XY}^{(\mathbf{x}_1^n,\mathbf{y}_1^n)}}{p_\rv{X}^{(\mathbf{x}_1^n)}}.
    \end{equation}
    Finally, for a conditional type $V_{\rv{Y}|\rv{X}}\in\set{P}_n(\set{Y}|\set{X})$ and a sequence $\mathbf{x}_1^n\in\set{X}^n$, we denote the set of sequences $\mathbf{y}_1^n$ such that $(\mathbf{x}_1^n,\mathbf{y}_1^n)$ is of conditional type $V_{\rv{Y}|\rv{X}}$ by $\set{T}_n(V_{\rv{Y}|\rv{X}},\mathbf{x}_1^n)$, \ie, 
    \begin{equation}
        \set{T}_n(V_{\rv{Y}|\rv{X}},\mathbf{x}_1^n)\defeq\left\{ \mathbf{y}_1^n\in\set{Y}^n \;\middle\vert\; p_{\rv{Y}|\rv{X}}^{(\mathbf{x}_1^n,\mathbf{y}_1^n)}(y|x) = V_{\rv{Y}|\rv{X}} \right\}.
    \end{equation}
    The set $\set{T}_n(V_{\rv{Y}|\rv{X}},\mathbf{x}_1^n)$ is also referred to as the $V$-shell of $\mathbf{x}_1^n$ \cite{csiszar2011information}.
\end{definition}
\begin{lemma}[{\cite[Lemma~2.5]{csiszar2011information}}]\label{lem:conditional:type:size}
For every $\mathbf{x}_1^n\in\set{X}^n$ and $V_{\rv{Y}|\rv{X}}\in\set{P}_n(\set{Y}|\set{X})$, it holds that
\begin{equation}
    n\cdot H_{}(\rv{Y}|\rv{X}) -\size{\set{X}}\cdot\size{\set{Y}}\cdot\log{(n+1)} \leq
    \log{\size*{\set{T}_n(V_{\rv{Y}|\rv{X}},\mathbf{x}_1^n)}} \leq 
    n\cdot H_{}(\rv{Y}|\rv{X})
\end{equation}
where 
$(\rv{X},\rv{Y})$ is distributed according to $ p_\rv{X}^{(\mathbf{x}_1^n)}\cdot V_{\rv{Y}|\rv{X}}$, and $H_{}(\rv{Y}|\rv{X})$ is the conditional entropy.
\end{lemma}
\begin{lemma}\label{lem:conditional:type:approx}
Let $\set{X}$ and $\set{Y}$ be two finite sets, and let $n$ be a positive integer.
Let $V_{\rv{Y}|\rv{X}}\in\set{P}(\set{Y}|\set{X})$ be a conditional pmf.
For each $n$-denominator type on $\set{X}$ $p_\rv{X}\in\set{P}_n(\set{X})$, there exists an approximation $V_{\rv{Y}|\rv{X}}^{(p_\rv{X})}\in\set{P}_n(\set{Y}|\set{X})$ of $V_{\rv{Y}|\rv{X}}$ such that
\begin{equation}\label{eq:conditional:type:approx}
    \abs*{p_\rv{X}(x)\cdot V_{\rv{Y}|\rv{X}}^{(p_\rv{X})}(y|x) - p_\rv{X}(x)\cdot V_{\rv{Y}|\rv{X}}(y|x)} \leq \frac{1}{n}
\end{equation}
for all $(x,y)\in\set{X}\times\set{Y}$.
\end{lemma}
\begin{proof}
See Appendix~\ref{app:continuity:proofs}.
\end{proof}
We shall also need the following continuity lemma of the Kullback-Leibler  divergence for proving Proposition~\ref{prop:EE-converse}.
\begin{lemma}\label{lem:D:continuity}
Let $\set{A}$ be a finite set, and $\xi\in(0,\frac{1}{e})$
For any two pmfs $p,p'\in\set{P}(\set{A})$ such that $\abs*{p(a)-p'(a)}\leq\xi$ { for all $a\in\set{A}$}, it holds for all $q\in\set{P}(\set{A})$, with $p\ll q$ and $p'\ll q$, that
\begin{equation}
\abs*{D\infdiv*{p}{q} - D\infdiv*{p'}{q}} \leq \xi\cdot \size*{\set{A}} \cdot \log{\frac{1}{q_{\min}}} + \size*{\set{A}}\cdot\xi\cdot\log{\frac{1}{\xi}}
\end{equation}
where $q_{\min}\defeq \min_{a\in\set{A}:\, q(a)>0} q(a)$.
\end{lemma}
\begin{proof}
See Appendix~\ref{app:continuity:proofs}.
\end{proof}
\begin{proof}[Proof of Proposition~\ref{prop:EE-converse}]
Again, we start with Proposition~\ref{prop:epsilon:NS:exact}.
By~\eqref{eq:epsilon:NS:exact:2}, we have
\begin{align}
    \epsilon^\fnc{NS}(\floor{e^{nr}})
    &= \adjustlimits \inf_{q_{\rvs{Y}_1^n}\in\set{P}(\set{Y}^n)} \sup_{p_{\rvs{X}_1^n}\in\set{P}(\set{X}^n)}\ 
    \sum_{\mathbf{x}_1^n\in\set{X}^n, \mathbf{y}_1^n\in\set{Y}^n}
    \left( p_{\rvs{X}_1^n}(\mathbf{x}_1^n)\cdot W_{\rv{Y}|\rv{X}}^{\tensor n}(\mathbf{y}_1^n|\mathbf{x}_1^n) - \floor{e^{nr}} \cdot p_{\rvs{X}_1^n}(\mathbf{x}_1^n)\cdot q_{\rvs{Y}_1^n}(\mathbf{y}_1^n)\right)_+ \\
    &\geq \adjustlimits \inf_{q_{\rvs{Y}_1^n}\in\set{P}(\set{Y}^n)} \sup_{p_{\rvs{X}_1^n}\in\set{P}(\set{X}^n)}\ 
    \sum_{\mathbf{x}_1^n\in\set{X}^n, \mathbf{y}_1^n\in\set{Y}^n}
    \left( p_{\rvs{X}_1^n}(\mathbf{x}_1^n)\cdot W_{\rv{Y}|\rv{X}}^{\tensor n}(\mathbf{y}_1^n|\mathbf{x}_1^n) - e^{nr} \cdot p_{\rvs{X}_1^n}(\mathbf{x}_1^n)\cdot q_{\rvs{Y}_1^n}(\mathbf{y}_1^n)\right)_+ .
\end{align}
Restricting the supremum to uniform distribution over sequences of a same type (see Definition~\ref{def:type}), we have
\begin{equation}
    \epsilon^\fnc{NS}(\floor{e^{nr}}) \geq \inf_{q_{\rvs{Y}_1^n}\in\set{P}(\set{Y}^n)} 
    \underbrace{\sup_{p_{\rv{X}}\in\set{P}_n(\set{X})}\ 
    \sum_{\mathbf{x}_1^n\in\set{T}_n(p_\rv{X})} \frac{1}{\size*{\set{T}_n(p_\rv{X})}}
    \sum_{\mathbf{y}_1^n\in\set{Y}^n}
    \left( W_{\rv{Y}|\rv{X}}^{\tensor n}(\mathbf{y}_1^n|\mathbf{x}_1^n) - e^{nr} \cdot q_{\rvs{Y}_1^n}(\mathbf{y}_1^n)\right)_+}_{\text{convex and permutation invariant \wrt } q_{\rvs{Y}_1^n}} .
\end{equation}
Since the argument after infimum is a supremum of convex functions of $q_{\rvs{Y}_1^n}$, and thus remains convex.
Observe that it is also permutation invariant.
Hence, there  exists\footnote{This holds when the infimum can be achieved, which we know to be the case since the target function is continuous and $\set{P}(\set{Y}^n)$ is a compact set.} a permutation invariant $q^\star_{\rvs{Y}_1^n}$ that achieves the infimum. 
By Lemma~\ref{lem:perm-inv-de-finetti}, there exists some distribution $\nu$ over $\set{P}(\set{Y})$ such that 
\begin{equation}
    q^\star_{\rvs{Y}_1^n} \leq\binom{n+|\set{Y}|-1}{n} \int \diff{\nu}(q_\rv{Y}) q_\rv{Y}^{\otimes n}
\leq (n+1)^{\size*{\set{Y}}-1} \cdot \underbrace{\int \diff\nu(q_\rv{Y})\ q_\rv{Y}^{\tensor n}}_{\defas \mu_{\rvs{Y}_1^n}}.
\end{equation}
Denoting the pmf after the multiplicative factor on the RHS by $\mu_{\rvs{Y}_1^n}$, we have 
\begin{equation}
    \epsilon^\fnc{NS}(\floor{e^{nr}}) \geq
    \sup_{p_{\rv{X}}\in\set{P}_n(\set{X})}\ 
    \sum_{\mathbf{x}_1^n\in\set{T}_n(p_\rv{X})} \frac{1}{\size*{\set{T}_n(p_\rv{X})}}
    \sum_{\mathbf{y}_1^n\in\set{Y}^n}
    \left( W_{\rv{Y}|\rv{X}}^{\tensor n}(\mathbf{y}_1^n|\mathbf{x}_1^n) - e^{nr} \cdot (n+1)^{\size*{\set{Y}}-1}\cdot\mu_{\rvs{Y}_1^n}(\mathbf{y}_1^n)\right)_+ \!.
\end{equation}

For \emph{any} $V^{(n)}_{\rv{Y}|\rv{X}}\in\set{P}_{n}(\set{Y}|\set{X})$, restricting the second summation to sequences $\mathbf{y_1^n}$ to the ones such that $(\mathbf{x}_1^n,\mathbf{y}_1^n)$ is of conditional type $V^{(n)}_{\rv{Y}|\rv{X}}$ (see Definition~\ref{def:conditonal:type}), we have 
\begin{equation}
    \epsilon^\fnc{NS}(\floor{e^{nr}}) \geq
    \sup_{p_{\rv{X}}\in\set{P}_n(\set{X})}\ 
    \sum_{\mathbf{x}_1^n\in\set{T}_n(p_\rv{X})} \frac{1}{\size*{\set{T}_n(p_\rv{X})}}
    \sum_{\mathbf{y}_1^n\in\set{T}_n(V^{(n)}_{\rv{Y}|\rv{X}},\mathbf{x}_1^n)}
    \left( W_{\rv{Y}|\rv{X}}^{\tensor n}(\mathbf{y}_1^n|\mathbf{x}_1^n) - e^{nr} \cdot (n+1)^{\size*{\set{Y}}-1}\cdot\mu_{\rvs{Y}_1^n}(\mathbf{y}_1^n)\right)_+ .
\end{equation}
By direct calculation, for each $\mathbf{x}_1^n\in\set{T}_n(p_\rv{X})$ and $\mathbf{y}_1^n\in\set{T}_n(V^{(n)}_{\rv{Y}|\rv{X}},\mathbf{x}_1^n)$ (or equivalently, $(\mathbf{x}_1^n,\mathbf{y}_1^n)\in\set{T}_n(p_\rv{X}\cdot V^{(n)}_{\rv{Y}|\rv{X}})$) it holds that \cite{csiszar2011information}
\begin{align}
\label{eq:W:conditional:type:identity}
W_{\rv{Y}|\rv{X}}^{\tensor n}(\mathbf{y}_1^n|\mathbf{x}_1^n)
&=\exp\left(-n\cdot H(\rv{Y}|\rv{X})\right) \cdot  \exp\left(-n\cdot D\infdiv*{p_\rv{X}\cdot V^{(n)}_{\rv{Y}|\rv{X}}}{p_\rv{X}\cdot W_{\rv{Y}|\rv{X}}} \right),\\
\label{eq:mu:conditional:type:identity}
\mu_{\rvs{Y}_1^n}(\mathbf{y}_1^n)
&=\int \diff\nu(q_\rv{Y})\ \cdot \exp\left(-n\cdot\left(D\infdiv*{p_\rv{X}\cdot V^{(n)}_{\rv{Y}|\rv{X}}}{p_\rv{X}\cdot q_\rv{Y}} + H(\rv{Y}|\rv{X})\right)\right) \\
&\leq \exp\left(-n\cdot H(\rv{Y}|\rv{X})\right) \cdot \max_{q_\rv{Y}\in\set{P}(\set{Y})} \exp\left(-n\cdot D\infdiv*{p_\rv{X}\cdot V^{(n)}_{\rv{Y}|\rv{X}}}{p_\rv{X}\cdot q_\rv{Y}}\right),
\end{align}
where the random variables $(\rv{X},\rv{Y})$ are distributed according to $p_\rv{X}\cdot V^{(n)}_{\rv{Y}|\rv{X}}$ for the conditional entropy $H(\rv{Y}|\rv{X})$ in  the above expressions.
This allows us to further write
\begin{equation}\begin{aligned}
    \epsilon^\fnc{NS}(\floor{e^{nr}}) \geq
    \sup_{p_{\rv{X}}\in\set{P}_n(\set{X})}\ 
    \sum_{\mathbf{x}_1^n\in\set{T}_n(p_\rv{X})} \frac{1}{\size*{\set{T}_n(p_\rv{X})}}
    \sum_{\mathbf{y}_1^n\in\set{T}_n(V^{(n)}_{\rv{Y}|\rv{X}},\mathbf{x}_1^n)}\!\!\!
    \Bigg( \!\exp\left(-n\!\cdot\!\left(D\infdiv*{p_\rv{X}\!\cdot\! V^{(n)}_{\rv{Y}|\rv{X}}}{p_\rv{X}\!\cdot\! W_{\rv{Y}|\rv{X}}} \!+\! H(\rv{Y}|\rv{X})\right)\!\right)\\
    - e^{nr} \cdot (n+1)^{\size*{\set{Y}}-1}\cdot \exp\left(-n\cdot H(\rv{Y}|\rv{X})\right) \cdot \max_{q_\rv{Y}\in\set{P}(\set{Y})} \exp\left(-n\cdot D\infdiv*{p_\rv{X}\cdot V^{(n)}_{\rv{Y}|\rv{X}}}{p_\rv{X}\cdot q_\rv{Y}}\right)\Bigg)_+ .
\end{aligned}\end{equation}
By Lemma~\ref{lem:conditional:type:size}~\cite[Lemma~2.5]{csiszar2011information}, we know 
\begin{equation}\label{eq:size:conditional:type}
\size*{\set{T}_n(V^{(n)}_{\rv{Y}|\rv{X}},\mathbf{x}_1^n)} \geq
\frac{1}{(n+1)^{\size*{\set{X}}\cdot\size*{\set{Y}}}} \cdot \exp\left(n\cdot H(\rv{Y}|\rv{X})\right),
\end{equation}
and we can therefore bound $\epsilon^\fnc{NS}$ as
\begin{equation}\begin{aligned}
    \epsilon^\fnc{NS}(\floor{e^{nr}}) \geq
    \sup_{p_{\rv{X}}\in\set{P}_n(\set{X})}\ 
    \frac{1}{(n+1)^{\size*{\set{X}}\cdot\size*{\set{Y}}}} \cdot
    \Bigg( \exp\left(-n\cdot D\infdiv*{p_\rv{X}\cdot V^{(n)}_{\rv{Y}|\rv{X}}}{p_\rv{X}\cdot W_{\rv{Y}|\rv{X}}}\right)\ldots\hspace{76pt}\\
    - e^{nr} \cdot (n+1)^{\size*{\set{Y}}-1}\cdot \max_{q_\rv{Y}\in\set{P}(\set{Y})} \exp\left(-n\cdot D\infdiv*{p_\rv{X}\cdot V^{(n)}_{\rv{Y}|\rv{X}}}{p_\rv{X}\cdot q_\rv{Y}}\right)\Bigg)_+ .
\end{aligned}\end{equation}
We can further restrict the maximization over $q_\rv{Y}$'s to a set of pmfs bounded away from zero.
In particular, for each $q_\rv{Y}\in\set{P}(\set{Y})$, we consider the pmf 
\begin{equation}\label{eq:bound:q:away:from:zero}
q_\rv{Y}^{(n)} \defeq \frac{q_\rv{Y}+\frac{1}{n}}{1+\frac{1}{n}\cdot\size*{\set{Y}}}.
\end{equation}
One can show by direct calculation that 
\begin{equation}
D\infdiv*{p_\rv{X}\cdot V_{\rv{Y}|\rv{X}}^{(n)}}{p_\rv{X}\times q_\rv{Y}^{(n)}} - D\infdiv*{p_\rv{X}\cdot V_{\rv{Y}|\rv{X}}^{(n)}}{p_\rv{X}\times q_\rv{Y}} \leq \frac{\size*{\set{Y}}}{n}.
\end{equation}
Note that $q_\rv{Y}^{(n)}$ is bounded away from zero, in particular $q_\rv{Y}^{(n)}(y) \geq \frac{1}{n+\size*{\set{Y}}}$ for all $y\in\set{Y}$.
Thus, by substituting $q_\rv{Y}$ by $q_\rv{Y}^{(\eta)}$ then relaxing the domain, we have
\begin{equation}\begin{aligned}
    \epsilon^\fnc{NS}(\floor{e^{nr}}) \geq
    \sup_{p_{\rv{X}}\in\set{P}_n(\set{X})}\ 
    \frac{1}{(n+1)^{\size*{\set{X}}\cdot\size*{\set{Y}}}} \cdot
    \Bigg( \exp\left(-n\cdot D\infdiv*{p_\rv{X}\cdot V^{(n)}_{\rv{Y}|\rv{X}}}{p_\rv{X}\cdot W_{\rv{Y}|\rv{X}}}\right)\ldots\hspace{76pt}\\
    - e^{nr} \cdot e^{\size*{\set{Y}}} \cdot (n+1)^{\size*{\set{Y}}-1} \cdot\sup_{q_\rv{Y}\in\set{P}(\set{Y}): q_\rv{Y} \geq \frac{1}{n+\size*{\set{Y}}}} \exp\left(-n\cdot D\infdiv*{p_\rv{X}\cdot V^{(n)}_{\rv{Y}|\rv{X}}}{p_\rv{X}\cdot q_\rv{Y}}\right)\Bigg)_+ .
\end{aligned}\end{equation}

Note that the above inequality holds for all conditional types $V^{(n)}_{\rv{Y}|\rv{X}}\in\set{P}_n(\set{Y}|\set{X})$, and thus for the supremum over all these types as well.
In particular, for each type $p_\rv{X}\in\set{P}_n(\set{X})$, we can restrict the supremum to approximations $V_{\rv{Y}|\rv{X}}^{(p_\rv{X})}\in\set{P}_n(\set{Y}|\set{X})$ of $V_{\rv{Y}|\rv{X}}$ (see Lemma~\ref{lem:conditional:type:approx}) where $V_{\rv{Y}|\rv{X}}\in\set{P}(\set{Y}|\set{X})$ can be any conditional pmf such that $V_{\rv{Y}|\rv{X}}\ll W_{\rv{Y}|\rv{X}}$, \ie,
\begin{equation}\label{eq:approx:V}\begin{aligned}
    \epsilon^\fnc{NS}(\floor{e^{nr}}) \geq
    \adjustlimits\sup_{p_{\rv{X}}\in\set{P}_n(\set{X})}
    \sup_{V_{\rv{Y}|\rv{X}}\in\set{P}(\set{Y}|\set{X}): V_{\rv{Y}|\rv{X}}\ll W_{\rv{Y}|\rv{X}}}\ 
    \frac{1}{(n+1)^{\size*{\set{X}}\cdot\size*{\set{Y}}}} \cdot
    \Bigg( \exp\left(\!-n\!\cdot\! D\infdiv*{p_\rv{X}\!\cdot\! V^{(p_\rv{X})}_{\rv{Y}|\rv{X}}}{p_\rv{X}\!\cdot\! W_{\rv{Y}|\rv{X}}}\right)\\
    - e^{nr} \cdot e^{\size*{\set{Y}}} \cdot (n+1)^{\size*{\set{Y}}-1}\cdot \sup_{q_\rv{Y}\in\set{P}(\set{Y}): q_\rv{Y} \geq \frac{1}{n+\size*{\set{Y}}}} \exp\left(-n\cdot D\infdiv*{p_\rv{X}\cdot V^{(p_\rv{X})}_{\rv{Y}|\rv{X}}}{p_\rv{X}\cdot q_\rv{Y}}\right)\Bigg)_+ .
\end{aligned}\end{equation}
By the construction of $V_{\rv{Y}|\rv{X}}^{(p_\rv{X})}$ in Lemma~\ref{lem:conditional:type:approx}, we know $p_\rv{X}\cdot V_{\rv{Y}|\rv{X}}^{(p_\rv{X})}\ll p_\rv{X}\cdot V_{\rv{Y}|\rv{X}}$.
This implies $p_\rv{X}\cdot V_{\rv{Y}|\rv{X}}^{(p_\rv{X})}\ll p_\rv{X}\cdot W_{\rv{Y}|\rv{X}}$.
Using~\eqref{eq:conditional:type:approx} and the continuity lemma on the Kullback-Leibler  divergence (Lemma~\ref{lem:D:continuity}), we have
\begin{align}
\abs*{D\infdiv*{p_\rv{X}\cdot V^{(p_\rv{X})}_{\rv{Y}|\rv{X}}}{p_\rv{X}\cdot W_{\rv{Y}|\rv{X}}} - D\infdiv*{p_\rv{X}\cdot V_{\rv{Y}|\rv{X}}}{p_\rv{X}\cdot W_{\rv{Y}|\rv{X}}} }
& \leq \frac{\size{\set{X}}\!\cdot\! \size{\set{Y}}}{n} \!\cdot\! \log{\frac{1}{W_{\min}}} + \size{\set{X}}\!\cdot\! \size{\set{Y}} \!\cdot\! \frac{1}{n}\log{n} \\
\abs*{D\infdiv*{p_\rv{X}\cdot V^{(p_\rv{X})}_{\rv{Y}|\rv{X}}}{p_\rv{X}\cdot q_\rv{Y}} - D\infdiv*{p_\rv{X}\cdot V_{\rv{Y}|\rv{X}}}{p_\rv{X}\cdot q_\rv{Y}}}
&\leq \frac{\size{\set{X}}\!\cdot\! \size{\set{Y}}}{n} \!\cdot\! \log{\left(n\!+\!\size*{\set{Y}}\right)} + \size{\set{X}} \!\cdot\! \size{\set{Y}} \!\cdot\! \frac{1}{n}\log{n}
\end{align}
for all integers $n\geq 3$, where we use the notation  $W_{\min}\defeq \min_{(x,y)\in\set{X}\times\set{Y}:\, W_{\rv{Y}|\rv{X}}(y|x)>0} W_{\rv{Y}|\rv{X}}(y|x)$.
Hence, we have
\begin{equation}\begin{aligned}
    \epsilon^\fnc{NS}(\floor{e^{nr}}) \geq
    \adjustlimits\sup_{p_{\rv{X}}\in\set{P}_n(\set{X})}
    \sup_{\genfrac{}{}{0pt}{}{V_{\rv{Y}|\rv{X}}\in\set{P}(\set{Y}|\set{X}):}{V_{\rv{Y}|\rv{X}}\ll W_{\rv{Y}|\rv{X}}}}
    \frac{1}{(n+1)^{\size*{\set{X}}\cdot\size*{\set{Y}}}} \cdot
    \Bigg( \!\!\!\left(\frac{W_{\min}}{n}\right)^{\size*{\set{X}}\cdot\size*{\set{Y}}}\hspace{-10pt}\cdot \exp\left(\!-n\!\cdot\! D\infdiv*{p_\rv{X}\!\cdot\! V_{\rv{Y}|\rv{X}}}{p_\rv{X}\!\cdot\! W_{\rv{Y}|\rv{X}}}\right)\\
    - e^{nr} \cdot e^{\size*{\set{Y}}} \cdot (n+1)^{\size*{\set{Y}}-1}\cdot \big(n\cdot(n+\size*{\set{Y}})\big)^{\size*{\set{X}}\cdot\size*{\set{Y}}} \cdot \sup_{q_\rv{Y}\in\set{P}(\set{Y}):q_\rv{Y}\geq\frac{1}{n+\size*{\set{Y}}}} \exp\left(-n\cdot D\infdiv*{p_\rv{X}\cdot V_{\rv{Y}|\rv{X}}}{p_\rv{X}\cdot q_\rv{Y}}\right)\Bigg)_+ .
\end{aligned}\end{equation}
Using a similar approach, we can further restrict the first supremum over $p_\rv{X}\in\set{P}_n(\set{X})$, to approximations $p^{(n)}_\rv{X}$ of $p_\rv{X}\in\set{P}(\set{X})$ (see Lemma~\ref{lem:type:approx}).
Using the fact that $\abs*{p^{(n)}_\rv{X}(x)-p_\rv{X}(x)}\leq\frac{1}{n}$ for all $x\in\set{X}$, we have for all $p_\rv{X}\in\set{P}(\set{X})$ that
\begin{align}
\abs*{D\infdiv*{p^{(n)}_\rv{X}\cdot V_{\rv{Y}|\rv{X}}}{p_\rv{X}\cdot W_{\rv{Y}|\rv{X}}} - D\infdiv*{p_\rv{X}\cdot V_{\rv{Y}|\rv{X}}}{p_\rv{X}\cdot W_{\rv{Y}|\rv{X}}} }
& \leq \frac{\size{\set{X}}\cdot \size{\set{Y}}}{n}\cdot \log{\frac{1}{W_{\min}}} \\
\abs*{D\infdiv*{p^{(n)}_\rv{X}\cdot V_{\rv{Y}|\rv{X}}}{p_\rv{X}\cdot q_\rv{Y}} - D\infdiv*{p_\rv{X}\cdot V_{\rv{Y}|\rv{X}}}{p_\rv{X}\cdot q_\rv{Y}}}
&\leq \frac{\size{\set{X}}\cdot \size{\set{Y}}}{n}\cdot \log{\left(n+\size*{\set{Y}}\right)}
\end{align}
for all integers $n\geq 3$.
Therefore, we can further rewrite
\begin{align}
    \label{EE:converse:temp:1}
    &\begin{aligned}
    \epsilon^\fnc{NS}(\floor{e^{nr}}) \geq
    \adjustlimits\sup_{p_{\rv{X}}\in\set{P}(\set{X})}
    \sup_{V_{\rv{Y}|\rv{X}}\in\set{P}(\set{Y}|\set{X})}\ 
    \frac{1}{(n+1)^{\size*{\set{X}}\cdot\size*{\set{Y}}}} \cdot
    \Bigg( \underbrace{\left(\frac{W_{\min}^2}{n}\right)^{\size*{\set{X}}\cdot\size*{\set{Y}}}}_{\defas A_n} \cdot \exp\left(-n\cdot D\infdiv*{p_\rv{X}\cdot V_{\rv{Y}|\rv{X}}}{p_\rv{X}\cdot W_{\rv{Y}|\rv{X}}}\right)\\
    - e^{nr} \cdot \underbrace{e^{\size*{\set{Y}}} \cdot (n+1)^{\size*{\set{Y}}-1}\cdot \big(n\cdot(n+\size*{\set{Y}})^2\big)^{\size*{\set{X}}\cdot\size*{\set{Y}}}}_{\defas B_n} \cdot \sup_{q_\rv{Y}\in\set{P}(\set{Y})} \exp\left(-n\cdot D\infdiv*{p_\rv{X}\cdot V_{\rv{Y}|\rv{X}}}{p_\rv{X}\cdot q_\rv{Y}}\right)\Bigg)_+ \end{aligned}\\
    \label{EE:converse:temp:2}
    & \begin{aligned} \phantom{\epsilon^\fnc{NS}(\floor{e^{nr}})} = 
    \adjustlimits\sup_{p_{\rv{X}}\in\set{P}(\set{X})}
    \sup_{V_{\rv{Y}|\rv{X}}\in\set{P}(\set{Y}|\set{X})}\ 
    \frac{1}{(n+1)^{\size*{\set{X}}\cdot\size*{\set{Y}}}} \cdot \bigg( A_n\cdot e^{-n\cdot D\infdiv*{p_\rv{X}\cdot V_{\rv{Y}|\rv{X}}}{p_\rv{X}\cdot W_{\rv{Y}|\rv{X}}}} - \ldots \hspace{66pt}\\ e^{nr}\cdot B_n \cdot e^{-n\cdot D\infdiv*{p_\rv{X}\cdot V_{\rv{Y}|\rv{X}}}{p_\rv{X}\cdot p_\rv{Y}}} \bigg)_+. \end{aligned}
\end{align}
Note that in~\eqref{EE:converse:temp:1} we also relaxed the domains of $V_{\rv{Y}|\rv{X}}$ and $q_\rv{Y}$ in the above expression.
The former can be done because all $V_{\rv{Y}|\rv{X}}\not\ll W_{\rv{Y}|\rv{X}}$ will either result in zero or can be equivalent to a $V_{\rv{Y}|\rv{X}}\ll W_{\rv{Y}|\rv{X}}$.
The latter can be done due to the minus sign before the maximization. In \eqref{EE:converse:temp:2}, 
we use the fact that $\inf_{q_\rv{Y}\in\set{P}(\set{Y})} D\infdiv*{p_\rv{X}\cdot V_{\rv{Y}|\rv{X}}}{p_\rv{X}\times q_\rv{Y}} = D\infdiv*{p_\rv{X}\cdot V_{\rv{Y}|\rv{X}}}{p_\rv{X}\times p_\rv{Y}}$ where 
\begin{equation}
p_\rv{Y}\defeq \sum_{x\in\set{X}} p_\rv{X}(x)\cdot V_{\rv{Y}|\rv{X}}(\cdot|x)
\end{equation}
is the output distribution corresponding to input $p_\rv{X}$ and channel $V_{\rv{Y}|\rv{X}}$.

Now, we would like to further bound the above expression by restricting $V_{\rv{Y}|\rv{X}}$ to the following subset of conditional pmfs
\begin{equation}
    \set{V}_n \defeq \left\{ V_{\rv{Y}|\rv{X}} \in\set{P}(\set{Y}|\set{X}) \;\middle\vert\;
    A_n\cdot e^{-n\cdot D\infdiv*{p_\rv{X}\cdot V_{\rv{Y}|\rv{X}}}{p_\rv{X}\cdot W_{\rv{Y}|\rv{X}}}} \geq \frac{e}{e-1}\cdot e^{nr}\cdot B_n \cdot e^{-n\cdot D\infdiv*{p_\rv{X}\cdot V_{\rv{Y}|\rv{X}}}{p_\rv{X}\times p_\rv{Y}}} \right\} .
\end{equation}
This allows us to write
\begin{equation}
    \epsilon^\fnc{NS}(\floor{e^{nr}}) \geq
    \adjustlimits\sup_{p_{\rv{X}}\in\set{P}(\set{X})}
    \sup_{V_{\rv{Y}|\rv{X}}\in\set{V}_n}\ 
    \frac{1}{(n+1)^{\size*{\set{X}}\cdot\size*{\set{Y}}}} \cdot
    \frac{1}{e}\cdot \left(\frac{W_{\min}^2}{n}\right)^{\size*{\set{X}}\cdot\size*{\set{Y}}} \cdot \exp\left(-n\cdot D\infdiv*{p_\rv{X}\cdot V_{\rv{Y}|\rv{X}}}{p_\rv{X}\cdot W_{\rv{Y}|\rv{X}}}\right),
\end{equation}
or equivalently
\begin{equation}\label{eq:EE-converse:temp}
    -\frac{1}{n}\log{\epsilon^\fnc{NS}(\floor{e^{nr}})}\le  - \frac{\size*{\set{X}}\cdot \size*{\set{Y}}}{n}\cdot\log{\frac{W^2_{\min}}{n\cdot (n+1)}} + \frac{1}{n} +
    \adjustlimits\inf_{p_{\rv{X}}\in\set{P}(\set{X})}
    \inf_{V_{\rv{Y}|\rv{X}}\in\set{V}_n}\ 
    D\infdiv*{p_\rv{X}\cdot V_{\rv{Y}|\rv{X}}}{p_\rv{X}\cdot W_{\rv{Y}|\rv{X}}}.
\end{equation}
Applying the method of Lagrange multipliers, the $V$-dependent term of the RHS of the above inequality can be written as 
\begin{align}
    &\hspace{12pt}\adjustlimits\inf_{p_{\rv{X}}\in\set{P}(\set{X})}
    \inf_{V_{\rv{Y}|\rv{X}}\in\set{V}_n}\ 
    D\infdiv*{p_\rv{X}\cdot V_{\rv{Y}|\rv{X}}}{p_\rv{X}\cdot W_{\rv{Y}|\rv{X}}} \nonumber \\
    & \begin{aligned} = \multiadjustlimits{
    \inf_{p_{\rv{X}}\in\set{P}(\set{X})},
    \inf_{V_{\rv{Y}|\rv{X}}\in\set{P}(\set{Y}|\set{X})},
    \sup_{\alpha\geq 0}} \
    D\infdiv*{p_\rv{X}\cdot V_{\rv{Y}|\rv{X}}}{p_\rv{X}\cdot W_{\rv{Y}|\rv{X}}} \ldots \hspace{191pt}\\
    -\alpha \cdot \left(D\infdiv*{p_\rv{X}\cdot V_{\rv{Y}|\rv{X}}}{p_\rv{X}\times p_\rv{Y}} - D\infdiv*{p_\rv{X}\cdot V_{\rv{Y}|\rv{X}}}{p_\rv{X}\times W_{\rv{Y}|\rv{X}}} -r +\frac{1}{n}\log{\frac{(e-1)A_n}{eB_n}}\right) \end{aligned}\\
    & \begin{aligned}=  \multiadjustlimits{
    \inf_{(p_{\rv{X}},V_{\rv{Y}|\rv{X}})\in\set{P}(\set{X})\times \set{P}(\set{Y}|\set{X})},
    \sup_{\alpha\geq 0}} \
    (1+\alpha)\cdot D\infdiv*{p_\rv{X}\cdot V_{\rv{Y}|\rv{X}}}{p_\rv{X}\cdot W_{\rv{Y}|\rv{X}}} 
    -\alpha\cdot D\infdiv*{p_\rv{X}\cdot V_{\rv{Y}|\rv{X}}}{p_\rv{X}\times p_\rv{Y}} \ldots \hspace{19pt} \\
    +\alpha\cdot \left(r-\frac{1}{n}\log{\frac{(e-1)A_n}{eB_n}}\right). \label{eq:D1}
    \end{aligned}
\end{align}
Observe that the above function is jointly convex in $(p_\rv{X}, V_{\rv{Y}|\rv{X}})$ for each fixed  $\alpha$ and linear in $\alpha$ for each fixed pair of $(p_\rv{X}, V_{\rv{Y}|\rv{X}})$; furthermore, the domain $\set{P}(\set{X})\times \set{P}(\set{Y}|\set{X})$ is compact convex and the domain $[0, \infty)$ is convex. 
Thus, Sion's minimax theorem~\cite{sion1958general} applies, and we have from \eqref{eq:D1}
\begin{align}
   &\adjustlimits\inf_{p_{\rv{X}}\in\set{P}(\set{X})}
    \inf_{V_{\rv{Y}|\rv{X}}\in\set{V}_n}\ 
    D\infdiv*{p_\rv{X}\cdot V_{\rv{Y}|\rv{X}}}{p_\rv{X}\cdot W_{\rv{Y}|\rv{X}}} \nonumber \\
     &\begin{aligned} = \multiadjustlimits{
    \sup_{\alpha\geq 0},
    \inf_{(p_{\rv{X}},V_{\rv{Y}|\rv{X}})\in\set{P}(\set{X})\times \set{P}(\set{Y}|\set{X})}}
    (1+\alpha)\cdot D\infdiv*{p_\rv{X}\cdot V_{\rv{Y}|\rv{X}}}{p_\rv{X}\cdot W_{\rv{Y}|\rv{X}}} \ldots \hspace{121pt} \\
    -\alpha\cdot D\infdiv*{p_\rv{X}\cdot V_{\rv{Y}|\rv{X}}}{p_\rv{X}\times p_\rv{Y}}
    +\alpha\cdot \left(r-\frac{1}{n}\log{\frac{(e-1)A_n}{eB_n}}\right),\end{aligned}\\
    &\begin{aligned}= \multiadjustlimits{
    \sup_{\alpha\geq 0},
    \inf_{(p_{\rv{X}},V_{\rv{Y}|\rv{X}})\in\set{P}(\set{X})\times \set{P}(\set{Y}|\set{X})},
    \sup_{q_\rv{Y}\in\set{P}(\set{Y})}} \
    (1+\alpha)\cdot D\infdiv*{p_\rv{X}\cdot V_{\rv{Y}|\rv{X}}}{p_\rv{X}\cdot W_{\rv{Y}|\rv{X}}} \ldots\hspace{81pt}\\
    -\alpha\cdot D\infdiv*{p_\rv{X}\cdot V_{\rv{Y}|\rv{X}}}{p_\rv{X}\times q_\rv{Y}} 
    +\alpha\cdot \left(r-\frac{1}{n}\log{\frac{(e-1)A_n}{eB_n}}\right),\end{aligned}\\
    &\begin{aligned}= \multiadjustlimits{
    \sup_{\alpha\geq 0},
    \sup_{q_\rv{Y}\in\set{P}(\set{Y})},
    \inf_{(p_{\rv{X}},V_{\rv{Y}|\rv{X}})\in\set{P}(\set{X})\times \set{P}(\set{Y}|\set{X})}}
    (1+\alpha)\cdot D\infdiv*{p_\rv{X}\cdot V_{\rv{Y}|\rv{X}}}{p_\rv{X}\cdot W_{\rv{Y}|\rv{X}}} \ldots\hspace{85pt}\\
    -\alpha\cdot D\infdiv*{p_\rv{X}\cdot V_{\rv{Y}|\rv{X}}}{p_\rv{X}\times q_\rv{Y}} 
    +\alpha\cdot \left(r-\frac{1}{n}\log{\frac{(e-1)A_n}{eB_n}}\right),\end{aligned}\label{eq:D2}
\end{align}
where we have used Sion's minimax theorem~\cite{sion1958general} again as the above function is concave in $q_\rv{Y}$ for each fixed $(p_\rv{X},V_{\rv{Y}})$, and jointly convex in $(p_\rv{X},V_{\rv{Y}})$ for each fixed $q_\rv{Y}$, and all of the  domains involved in this optimization are convex and compact.
Using the variational formulation of the Rényi divergence~\cite[Theorem 30]{vanErven2014Jun} and starting from \eqref{eq:D2}, we further have the following  
\begin{align}
&\adjustlimits\inf_{p_{\rv{X}}\in\set{P}(\set{X})}
    \inf_{V_{\rv{Y}|\rv{X}}\in\set{V}_n}\ 
    D\infdiv*{p_\rv{X}\cdot V_{\rv{Y}|\rv{X}}}{p_\rv{X}\cdot W_{\rv{Y}|\rv{X}}} \nonumber \\
    &= \multiadjustlimits{
    \sup_{\alpha\geq 0},
    \sup_{q_\rv{Y}\in\set{P}(\set{Y})},
    \inf_{p_\rv{X}\in\set{P}(\set{X})}}
    -\alpha\cdot  \mathbb{E}_{x\sim p_\rv{X}}\left[ D_{1+\alpha}\infdiv*{W_{\rv{Y}|\rv{X}}(\cdot|x)}{ q_\rv{Y}}\right] +\alpha\cdot \left(r-\frac{1}{n}\log{\frac{(e-1)A_n}{eB_n}}\right)
    \\&=\multiadjustlimits{
    \sup_{\alpha\geq 0},
    \sup_{q_\rv{Y}\in\set{P}(\set{Y})},
    \inf_{p_\rv{X}\in\set{P}(\set{X})}}
    -\alpha\cdot D_{1+\alpha}\infdiv*{p_\rv{X}\cdot W_{\rv{Y}|\rv{X}}}{p_\rv{X}\times q_\rv{Y}} +\alpha\cdot \left(r-\frac{1}{n}\log{\frac{(e-1)A_n}{eB_n}}\right)
\end{align}
where we use \cite[Proposition 1]{csiszar1995generalized} in the last equality. 
Combining above with~\eqref{eq:EE-converse:temp}, we have
\begin{equation}
    -\frac{1}{n}\log{\epsilon^\fnc{NS}(\floor{e^{nr}})} \leq
    \multiadjustlimits{
    \sup_{\alpha\geq 0},
    \sup_{q_\rv{Y}\in\set{P}(\set{Y})},
    \inf_{p_\rv{X}\in\set{P}(\set{X})}}
    \alpha\cdot \left(r-D_{1+\alpha}\infdiv*{p_\rv{X}\cdot W_{\rv{Y}|\rv{X}}}{p_\rv{X}\times q_\rv{Y}}+g_n\right) + f_n
\end{equation}
where
\begin{align}
f_n &\defeq \frac{\size*{\set{X}}\cdot \size*{\set{Y}}}{n}\cdot\log{\tfrac{n\cdot (n+1)}{W^2_{\min}}} + \frac{1}{n} = O\left(\frac{\log{n}}{n}\right),\\
g_n &\defeq \frac{1}{n}\log{\tfrac{eB_n}{(e-1)A_n}}
= \frac{1}{n}\log{\tfrac{e\cdot e^{\size*{\set{Y}}} \cdot (n+1)^{\size*{\set{Y}}-1}\cdot \big(n\cdot(n+\size*{\set{Y}})\big)^{2\size*{\set{X}}\cdot\size*{\set{Y}}}}{(e-1)\cdot W_{\min}^{2\size*{\set{X}}\cdot\size*{\set{Y}}}}}=O\left(\frac{\log{n}}{n}\right). \qedhere
\end{align}
\end{proof}
From Propositions \ref{prop:EE-Ach} and \ref{prop:EE-converse}, we deduce the error exponent for non-signaling channel simulation. 
\begin{theorem}\label{thm:NS:EE}
Let $W_{\rv{Y}|\rv{X}}\in\set{P}(\set{Y}|\set{X})$ be a channel.
For all $r> 0$, it holds that 
\begin{equation}
    \lim_{n\to \infty} -\frac{1}{n} \log\epsilon^\fnc{NS}(\floor{e^{nr}},W_{\rv{Y}|\rv{X}}^{\tensor n}) =
    \sup_{\alpha\geq 0} \; \alpha\cdot\left( r - \adjustlimits\sup_{p_\rv{X}\in\set{P}(\set{X})} \inf_{q_\rv{Y}\in\set{P}(\set{Y})} D_{\alpha+1}\infdiv*{p_\rv{X}\cdot W_{\rv{Y}|\rv{X}}}{p_\rv{X}\times q_\rv{Y}} \right).
\end{equation}
\end{theorem}
\section{Strong Converse Exponent for Non-Signaling Channel Simulation}\label{sec:SCE}
In this section, we study how fast $\epsilon^\fnc{NS}$ converges to $1$ for simulating asymptotically many copies of a channel given a fixed communication rate $r>0$.
In particular, we would like to study the strong converse exponent $-\frac{1}{n}\log{\left(1-\epsilon^\fnc{NS}(\floor{e^{nr}},W_{\rv{Y}|\rv{X}}^{\tensor n})\right)}$ as $n$ tends to infinity.
Proposition~\ref{prop:SCE:Converse} below shows that $\epsilon^\fnc{NS}$ converges to $1$ exponentially fast for any finite blocklength $n$ when $r<I_{1}(W_{\rv{Y}|\rv{X}})$.
Later, Proposition~\ref{prop:SCE:achievability} shows that the established convergence rate (as known as the strong converse exponent) is asymptotically tight.

\subsection{Converse for the Strong Converse Exponent}\label{sec:SCE:Converse}
In the following, we would like to prove the following lower bound on the strong converse exponent.
\begin{proposition}[Converse for the Strong Converse Exponents]\label{prop:SCE:Converse}
Let $W_{\rv{Y}|\rv{X}}\in\set{P}(\set{Y}|\set{X})$ be a channel.
For all $r> 0$, it holds that 
\begin{equation}\label{eq:EE-Converse-finite}
     \frac{1}{n}\log{\left(1-\epsilon^\fnc{NS}(\floor{e^{nr}},W_{\rv{Y}|\rv{X}}^{\tensor n})\right)} \leq
    - \sup_{\alpha\in[0,1]} \; \alpha\cdot\left( -r + \adjustlimits\sup_{p_\rv{X}\in\set{P}(\set{X})} \inf_{q_\rv{Y}\in\set{P}(\set{Y})} D_{1-\alpha}\infdiv*{p_\rv{X}\cdot W_{\rv{Y}|\rv{X}}}{p_\rv{X}\times q_\rv{Y}} \right).
\end{equation}
\end{proposition}
\begin{remark}
By following the argument in Remark~\ref{remark:EE}, the established strong converse exponent for non-signaling channel simulation, \ie, $\sup_{\alpha\in [0,1]} \alpha\cdot \left( I_{\alpha}(W_{\rv{Y}|\rv{X}}) - r \right)  $, is positive if and only if $r < I_{1}(W_{\rv{Y}|\rv{X}})$.
Hence, $\epsilon^\fnc{NS}$ converges to $1$ exponentially fast as $r < I_{1}(W_{\rv{Y}|\rv{X}})$.
\end{remark}
\begin{proof}
Starting from Proposition~\ref{prop:epsilon:NS:exact}, we have 
\begin{align}
    &\hspace{14pt}1-\epsilon^\fnc{NS}(\floor{e^{nr}},W_{\rv{Y}|\rv{X}}^{\tensor n})\nonumber\\
    &= 1\!-\!\adjustlimits \inf_{q_{\rvs{Y}_1^n}\in\set{P}(\set{Y}^n)} \sup_{p_{\rvs{X}_1^n}\in\set{P}(\set{X}^n)}\ 
    \sum_{\mathbf{x}_1^n\in\set{X}^n, \mathbf{y}_1^n\in\set{Y}^n}
    \Big( p_{\rvs{X}_1^n}(\mathbf{x}_1^n)\!\cdot\! W_{\rv{Y}|\rv{X}}^{\tensor n}(\mathbf{y}_1^n|\mathbf{x}_1^n) - \floor{e^{nr}} \!\cdot
    \! p_{\rvs{X}_1^n}(\mathbf{x}_1^n) \!\cdot\! q_{\rvs{Y}_1^n}(\mathbf{y}_1^n)\Big)_+, \!\\
    &=\adjustlimits \sup_{q_{\rvs{Y}_1^n}\in\set{P}(\set{Y}^n)} \inf_{p_{\rvs{X}_1^n}\in\set{P}(\set{X}^n)}\ 
    \sum_{\mathbf{x}_1^n\in\set{X}_1^n} p_{\rvs{X}_1^n}(\mathbf{x}_1^n)\cdot \sum_{\mathbf{y}_1^n}
    \min\left\{  W_{\rv{Y}|\rv{X}}^{\tensor n}(\mathbf{y}_1^n|\mathbf{x}_1^n),\, \floor{e^{nr}} \cdot q_{\rvs{Y}_1^n}(\mathbf{y}_1^n)\right\} \label{eq:SEE-Converse:1}\\
    &\leq \adjustlimits \inf_{p_\rv{X}\in\set{P}(\set{X})} \sup_{q_{\rvs{Y}_1^n}\in\set{P}(\set{Y}^n)} \ 
    \sum_{\mathbf{x}_1^n\in\set{X}_1^n} p_\rv{X}^{\tensor n}(\mathbf{x}_1^n)\cdot \sum_{\mathbf{y}_1^n}
    \min\left\{  W_{\rv{Y}|\rv{X}}^{\tensor n}(\mathbf{y}_1^n|\mathbf{x}_1^n),\, e^{nr} \cdot q_{\rvs{Y}_1^n}(\mathbf{y}_1^n)\right\} \label{eq:SEE-Converse:2}
\end{align}
where in~\eqref{eq:SEE-Converse:1}, we have applied the fact that $1-\sum_{a\in\set{A}} \left(p(a)-f(a)\right)_+ = \sum_{a\in\set{A}} \min\left\{p(a), f(a)\right\}$ for any pmf $p\in\set{P}(\set{A})$ and function $f:\set{A}\to\reals$;
and where in~\eqref{eq:SEE-Converse:2}, we swap the supremum inside, restrict the domain of the infimum to product distributions, and use the fact that $\floor{e^{nr}}\leq e^{nr}$.
Since $\min\left\{a,b\right\}\leq a^\alpha\cdot b^{1-\alpha}$ for any $a,b\geq 0$ and $\alpha\in[0,1]$, we can further upper bound the above expression as follows
\begin{align}
    \text{Above} &\leq 
    \adjustlimits \inf_{p_\rv{X}\in\set{P}(\set{X})} \sup_{q_{\rvs{Y}_1^n}\in\set{P}(\set{Y}^n)} \ 
    \sum_{\mathbf{x}_1^n\in\set{X}_1^n} p_\rv{X}^{\tensor n}(\mathbf{x}_1^n)\cdot \sum_{\mathbf{y}_1^n}
    \Big(W_{\rv{Y}|\rv{X}}^{\tensor n}(\mathbf{y}_1^n|\mathbf{x}_1^n)\Big)^{1-\alpha}\cdot \Big(e^{nr} \cdot q_{\rvs{Y}_1^n}(\mathbf{y}_1^n)\Big)^{\alpha}, \\
    &= e^{\alpha\cdot n\cdot r} \cdot \adjustlimits \inf_{p_\rv{X}\in\set{P}(\set{X})} \sup_{q_{\rvs{Y}_1^n}\in\set{P}(\set{Y}^n)} \ 
    \sum_{\mathbf{x}_1^n\in\set{X}_1^n, \mathbf{y}_1^n\in\set{Y}_1^n} \Big(p_\rv{X}^{\tensor n}(\mathbf{x}_1^n)\cdot W_{\rv{Y}|\rv{X}}^{\tensor n}(\mathbf{y}_1^n|\mathbf{x}_1^n)\Big)^{1-\alpha}\cdot \Big(p_\rv{X}^{\tensor n}(\mathbf{x}_1^n)\cdot q_{\rvs{Y}_1^n}(\mathbf{y}_1^n)\Big)^{\alpha}, \\
    &= \exp\left(\alpha\cdot n\cdot r
    -\alpha\cdot \adjustlimits\sup_{p_\rv{X}\in\set{P}(\set{X})} \inf_{q_{\rvs{Y}_1^n}\in\set{P}(\set{Y}^n)} D_{1-\alpha}\infdiv*{p_\rv{X}^{\tensor n}\cdot W_{\rv{Y}|\rv{X}}^{\tensor n}}{p_\rv{X}^{\tensor n}\times q_{\rvs{Y}_1^n}}\right).
\end{align}
By the additivity property of the mutual information of order $1-\alpha$~\cite{arimoto1977information}, \ie, 
\begin{equation}
    \inf_{q_{\rvs{Y}_1^n}\in\set{P}(\set{Y}^n)} D_{1-\alpha}\infdiv*{p_\rv{X}^{\tensor n}\cdot W_{\rv{Y}|\rv{X}}^{\tensor n}}{p_\rv{X}^{\tensor n}\times q_{\rvs{Y}_1^n}}
    = n\cdot \inf_{q_\rv{Y}\in\set{P}(\set{Y})} D_{1-\alpha}\infdiv*{p_\rv{X}\cdot W_{\rv{Y}|\rv{X}}}{p_\rv{X}\times q_\rv{Y}},
\end{equation}
we have 
\begin{align}
    1-\epsilon^\fnc{NS}(\floor{e^{nr}},W_{\rv{Y}|\rv{X}}^{\tensor n}) &\leq 
    \exp\left(-\alpha\cdot n\cdot \left(-r+
    \adjustlimits\sup_{p_\rv{X}\in\set{P}(\set{X})}\inf_{q_\rv{Y}\in\set{P}(\set{Y})}
    D_{1-\alpha}\infdiv*{p_\rv{X}\cdot W_{\rv{Y}|\rv{X}}}{p_\rv{X}\times q_\rv{Y}}
    \right)\right).
\end{align}
Finally,~\eqref{eq:EE-Converse-finite} is obtained by taking the logarithm and dividing by $n$ on both sides of the above inequality.
\end{proof}
\subsection{Achievability for the Strong Converse Exponent}\label{sec:SCE:achievability}
In the following, we move to prove the following upper bound on the strong converse exponent.
\begin{proposition}\label{prop:SCE:achievability}
Let $W_{\rv{Y}|\rv{X}}\in\set{P}(\set{Y}|\set{X})$ be a channel.
For all $r>0$, it holds that 
\begin{equation}\begin{aligned}
    \frac{1}{n}\log{\left(1\!-\!\epsilon^\fnc{NS}(\floor{e^{nr}},W_{\rv{Y}|\rv{X}}^{\tensor n})\right)} \geq -\sup_{\alpha\in[0,1]}
    (1\!-\!\alpha)\cdot 
    \left(-r_n \!+\! \adjustlimits\sup_{p_\rv{X}\in\set{P}(\set{X})}\inf_{q_\rv{Y}\in\set{P}(\set{Y})} D_{\alpha}\infdiv*{p_\rv{X}\!\cdot\! W_{\rv{Y}|\rv{X}}}{p_\rv{X}\!\times\! q_\rv{Y}}\right) \\
    + \alpha\cdot\tilde{g}_n - \tilde{f}_n
\end{aligned}\end{equation}
for all integer $n\geq 3\cdot\size*{\set{X}}$, for some sequences $\{\tilde{f}_n\}_n$ and $\{\tilde{g}_n\}_n$ such that $\tilde{f}_n=O\left(\frac{\log{n}}{n}\right)$ and $\tilde{g}_n=O\left(\frac{\log{n}}{n}\right)$.
Here, $r_n \defeq \frac{1}{n}\log{\floor{e^{nr}}}$ has already been defined in Proposition~\ref{prop:EE-Ach}.
In particular, as $n$ tends to infinity, we have
\begin{align}
    \liminf_{n\to \infty } \frac{1}{n}\log{\left(1\!-\!\epsilon^\fnc{NS}(\floor{e^{nr}},W_{\rv{Y}|\rv{X}}^{\tensor n})\right)}  &\geq -\sup_{\alpha\in[0,1]}\;
    (1\!-\!\alpha)\cdot 
    \left(-r \!+\! \adjustlimits\sup_{p_\rv{X}\in\set{P}(\set{X})}\inf_{q_\rv{Y}\in\set{P}(\set{Y})} D_{\alpha}\infdiv*{p_\rv{X}\!\cdot\! W_{\rv{Y}|\rv{X}}}{p_\rv{X}\!\times\! q_\rv{Y}}\right).
\end{align}
\end{proposition}
We shall need the following continuation lemma on the mutual information for the proof of Proposition~\ref{prop:SCE:achievability}.
\begin{lemma}\label{lem:I:continuity}
Let $\set{X}$ and $\set{Y}$ be two finite sets.
Let $p_\rv{X}\in\set{P}(\set{X})$ be a pmf on $\set{X}$, and let $V_{\rv{Y}|\rv{X}}\in\set{P}(\set{Y}|\set{X})$ be a channel from $\set{X}$ to $\set{Y}$.
\begin{enumerate}
\item For any $\tilde{V}_{\rv{Y}|\rv{X}}\in\set{P}(\set{Y}|\set{X})$ such that $\abs{p_\rv{X}(x)\cdot\tilde{V}_{\rv{Y}|\rv{X}}(y|x)-p_\rv{X}(x)\cdot V_{\rv{Y}|\rv{X}}(y|x)}\leq\xi\ \forall (x,y)\in\set{X}\times\set{Y}$ where $0 < \xi\leq \frac{1}{\size*{\set{X}}\cdot e}$, it hold that
    \begin{equation}\label{eq:I:continuity:1}
    \abs*{D\infdiv*{p_\rv{X}\cdot V_{\rv{Y}|\rv{X}}}{p_\rv{X}\times p_\rv{Y}} - D\infdiv*{p_\rv{X}\cdot \tilde{V}_{\rv{Y}|\rv{X}}}{p_\rv{X}\times \tilde{p}_\rv{Y}}} \leq \xi\cdot \size*{\set{X}}\cdot\size*{\set{Y}}\cdot \left(\log{\frac{1}{\xi}} + \log{\frac{1}{\xi\cdot \size*{\set{X}}}} \right),
    \end{equation}
    where $p_\rv{Y}$ and $\tilde{p}_\rv{Y}$ are the induced output distributions of the channels $V_{\rv{Y}|\rv{X}}$ and $\tilde{V}_{\rv{Y}|\rv{X}}$ (with the same input distribution $p_\rv{X}$), respectively.
    \item For any $\tilde{p}_\rv{X}\in\set{P}(\set{X})$ such that $\abs{\tilde{p}_\rv{X}(x)-p_\rv{X}(x)}\leq\xi\ \forall x\in\set{X}$ where $0 < \xi\leq \frac{1}{\size*{\set{X}}\cdot e}$, it hold that
    \begin{equation}\label{eq:I:continuity:2}
    \abs*{D\infdiv*{p_\rv{X}\cdot V_{\rv{Y}|\rv{X}}}{p_\rv{X}\times p_\rv{Y}} - D\infdiv*{\tilde{p}_\rv{X}\cdot V_{\rv{Y}|\rv{X}}}{\tilde{p}_\rv{X}\times \tilde{p}_\rv{Y}}} \leq \xi\cdot\size*{\set{X}}\cdot\log{\size*{\set{Y}}} + \xi\cdot \size*{\set{X}}\cdot\size*{\set{Y}}\cdot \log{\frac{1}{\xi\cdot \size*{\set{X}}}}
    \end{equation}
    where $p_\rv{Y}$ and $\tilde{p}_\rv{Y}$ are the induced output distributions corresponding to the input distributions $p_\rv{X}$ and $\tilde{p}_\rv{X}$ (with the same channel $V_{\rv{Y}|\rv{X}}$), respectively.
\end{enumerate}
\end{lemma}
\begin{proof}
See Appendix~\ref{app:continuity:proofs}.
\end{proof}
\begin{proof}[Proof of Proposition~\ref{prop:SCE:achievability}]
We start with~\eqref{eq:SEE-Converse:1}, \ie, 
\begin{align}
    1-\epsilon^\fnc{NS}(\floor{e^{nr}},W_{\rv{Y}|\rv{X}}^{\tensor n})
    &= \adjustlimits \sup_{q_{\rvs{Y}_1^n}\in\set{P}(\set{Y}^n)} \inf_{p_{\rvs{X}_1^n}\in\set{P}(\set{X}^n)}\ 
    \sum_{\mathbf{x}_1^n\in\set{X}_1^n} p_{\rvs{X}_1^n}(\mathbf{x}_1^n)\cdot \sum_{\mathbf{y}_1^n} \min\left\{  W_{\rv{Y}|\rv{X}}^{\tensor n}(\mathbf{y}_1^n|\mathbf{x}_1^n),\, \floor{e^{nr}}\cdot q_{\rvs{Y}_1^n}(\mathbf{y}_1^n) \right\} \tag{\ref{eq:SEE-Converse:1}, repeated}.
\end{align}
Since the function to be optimized above is linear in $p_{\rvs{X}_1^n}$ and concave in $q_{\rvs{Y}_1^n}$, and the sets $\set{P}(\set{X}^n)$ and $\set{P}(\set{Y}^n)$ are convex compact, by Sion's minimax theorem~\cite{sion1958general}, we have 
\begin{align}
    &1-\epsilon^\fnc{NS}(\floor{e^{nr}},W_{\rv{Y}|\rv{X}}^{\tensor n}) \nonumber
   \\ &= \adjustlimits \inf_{p_{\rvs{X}_1^n}\in\set{P}(\set{X}^n)} \sup_{q_{\rvs{Y}_1^n}\in\set{P}(\set{Y}^n)}\ 
    \sum_{\mathbf{x}_1^n\in\set{X}_1^n} p_{\rvs{X}_1^n}(\mathbf{x}_1^n)\cdot \sum_{\mathbf{y}_1^n\in\set{Y}^n} \min\left\{  W_{\rv{Y}|\rv{X}}^{\tensor n}(\mathbf{y}_1^n|\mathbf{x}_1^n),\, \floor{e^{nr}}\cdot q_{\rvs{Y}_1^n}(\mathbf{y}_1^n) \right\}\\
    \label{eq:SEE-Achievability:1}
    &\geq \inf_{p_{\rvs{X}_1^n}\in\set{P}(\set{X}^n)} \underbrace{\sup_{q_\rv{Y}\in\set{P}(\set{Y})}\ 
    \sum_{\mathbf{x}_1^n\in\set{X}_1^n} p_{\rvs{X}_1^n}(\mathbf{x}_1^n)\cdot \sum_{\mathbf{y}_1^n\in\set{Y}^n} \min\left\{  W_{\rv{Y}|\rv{X}}^{\tensor n}(\mathbf{y}_1^n|\mathbf{x}_1^n),\, \floor{e^{nr}}\cdot q_\rv{Y}^{\tensor n}(\mathbf{y}_1^n) \right\}}_\text{convex \wrt $p_{\rvs{X}_1^n}$ and invariant under all permutations of $\rv{X}_1,\ldots,\rv{X}_n$}
\end{align}
where in~\eqref{eq:SEE-Achievability:1} we restricted the supremum to i.i.d. pmfs.
In addition, note that the expression after the infimum is a convex function of $p_{\rvs{X}_1^n}$, and is invariant under all permutations of $\rv{X}_1,\ldots,\rv{X}_n$.
Thus, there exists some permutation invariant $p^\star_{\rvs{X}_1^n}$ that achieves the infimum.
As every length-$n$ permutation invariant pmf can be written as convex combinations of uniform distributions over sequences of same types (see Definition~\ref{def:type} for notation), \ie
\begin{equation}
    p^\star_{\rvs{X}_1^n}(\mathbf{x}_1^n) = \sum_{p_\rv{X}\in\set{P}_n(\set{X})} \alpha_{p_\rv{X}}\cdot \frac{1}{\size*{\set{T}_n(p_\rv{X})}}\cdot \mathbbm{1}\left\{\mathbf{x}_1^n\in\set{T}_n(p_\rv{X})\right\} \quad\forall \mathbf{x}_1^n\in\set{X}^n,
\end{equation}
for some pmf $\{\alpha_{p_\rv{X}}\}_{p_\rv{X}\in\set{P}_n(\set{X})}$ on $\set{P}_n(\set{X})$,
there exists some type $p_\rv{X}\in\set{P}_n(\set{X})$ such that
\begin{equation}
	p^\star_{\rvs{X}_1^n} \geq \frac{1}{\size*{\set{P}_n(\set{X})}}\cdot \frac{1}{\size*{\set{T}_n(p_\rv{X})}}\cdot \mathbbm{1}\left\{\mathbf{x}_1^n\in\set{T}_n(p_\rv{X})\right\}
	\ge  \frac{1}{(n+1)^{\size*{\set{X}}-1}} \cdot \frac{1}{\size*{\set{T}_n(p_\rv{X})}}\cdot \mathbbm{1}\left\{\mathbf{x}_1^n\in\set{T}_n(p_\rv{X})\right\},
\end{equation}
where we used $\size*{\set{P}_n(\set{X})}=\binom{n+\size*{\set{X}}-1}{n} \le (n+1)^{\size*{\set{X}}-1}$. 
This allows us to lower bound~\eqref{eq:SEE-Achievability:1} by
\begin{equation}\begin{aligned}
    1-\epsilon^\fnc{NS}(\floor{e^{nr}},W_{\rv{Y}|\rv{X}}^{\tensor n})
    \geq \frac{1}{(n+1)^{\size*{\set{X}}-1}} \cdot
    \adjustlimits \inf_{p_\rv{X}\in\set{P}_n(\set{X})} \sup_{q_\rv{Y}\in\set{P}(\set{Y})}
    \sum_{\mathbf{x}_1^n\in\set{T}_n(p_\rv{X})}\frac{1}{\size*{\set{T}_n(p_\rv{X})}} \cdot \ldots \hspace{82pt}\\
    \sum_{\mathbf{y}_1^n\in\set{Y}^n}\min\left\{  W_{\rv{Y}|\rv{X}}^{\tensor n}(\mathbf{y}_1^n|\mathbf{x}_1^n),\, \floor{e^{nr}}\cdot q_\rv{Y}^{\tensor n}(\mathbf{y}_1^n) \right\}
\end{aligned}\end{equation}

Now, we apply the method of types in a similar fashion as in Section~\ref{sec:EE-converse}.
Let $V_{\rv{Y}|\rv{X}}^{(n)}\in\set{P}_n(\set{Y}|\set{X})$ be an arbitrary conditional type (see Definition~\ref{def:conditonal:type}).
Restricting the summation over $\mathbf{y}_1^n$ to the ones such that $(\mathbf{x}_1^n,\mathbf{y}_1^n)$ is of conditional type $V^{(n)}_{\rv{Y}|\rv{X}}$ we have 
\begin{align}
&\begin{aligned}
    1-\epsilon^\fnc{NS}(\floor{e^{nr}},W_{\rv{Y}|\rv{X}}^{\tensor n})
    \geq \frac{1}{(n+1)^{\size*{\set{X}}-1}} \cdot
    \multiadjustlimits{
        \inf_{p_\rv{X}\in\set{P}_n(\set{X})},
        \sup_{q_\rv{Y}\in\set{P}(\set{Y})},
        \sup_{V_{\rv{Y}|\rv{X}}^{(n)}\in\set{P}_n(\set{Y}|\set{X})}}\ 
    \sum_{\mathbf{x}_1^n\in\set{T}_n(p_\rv{X})}\frac{1}{\size*{\set{T}_n(p_\rv{X})}} \cdot \ldots \hspace{19pt}\\
    \sum_{\mathbf{y}_1^n\in\set{T}_n(V^{(n)}_{\rv{Y}|\rv{X}},\mathbf{x}_1^n)} 
    \min\Bigg\{ \exp\left(-n\cdot\left(D\infdiv*{p_\rv{X}\cdot V^{(n)}_{\rv{Y}|\rv{X}}}{p_\rv{X}\cdot W_{\rv{Y}|\rv{X}}} + H(\rv{Y}|\rv{X})\right) \right),\,\ldots\\
    \floor{e^{nr}}\cdot \exp\left(-n\cdot\left(D\infdiv*{p_\rv{X}\cdot V^{(n)}_{\rv{Y}|\rv{X}}}{p_\rv{X}\times q_\rv{Y}} + H(\rv{Y}|\rv{X})\right)\right) \Bigg\}
\end{aligned}\\
& \begin{aligned}
    \phantom{1-\epsilon^\fnc{NS}(\floor{e^{nr}},W_{\rv{Y}|\rv{X}}^{\tensor n})}=\frac{1}{(n+1)^{\size*{\set{X}}-1}} \cdot
    \multiadjustlimits{
        \inf_{p_\rv{X}\in\set{P}_n(\set{X})},
        \sup_{q_\rv{Y}\in\set{P}(\set{Y})},
        \sup_{V_{\rv{Y}|\rv{X}}^{(n)}\in\set{P}_n(\set{Y}|\set{X})}}\ 
    \sum_{\mathbf{y}_1^n\in\set{T}_n(V^{(n)}_{\rv{Y}|\rv{X}},\mathbf{x}_1^n)} \ldots \hspace{46pt}\\
    \min\Bigg\{ \exp\left(-n\cdot\left(D\infdiv*{p_\rv{X}\cdot V^{(n)}_{\rv{Y}|\rv{X}}}{p_\rv{X}\cdot W_{\rv{Y}|\rv{X}}} + H(\rv{Y}|\rv{X})\right) \right),\, \ldots \\
    \floor{e^{nr}}\cdot \exp\left(-n\cdot\left(D\infdiv*{p_\rv{X}\cdot V^{(n)}_{\rv{Y}|\rv{X}}}{p_\rv{X}\times q_\rv{Y}} + H(\rv{Y}|\rv{X})\right)\right) \Bigg\}
\end{aligned}
\end{align}
where we have used~\eqref{eq:W:conditional:type:identity} and the fact that (also see~\eqref{eq:mu:conditional:type:identity})
\begin{equation}
	q_\rv{Y}^{\tensor n}(\mathbf{y}_1^n) = \exp\left(-n\cdot\left(D\infdiv*{p_\rv{X}\cdot V^{(n)}_{\rv{Y}|\rv{X}}}{p_\rv{X}\cdot q_\rv{Y}} + H(\rv{Y}|\rv{X})\right)\right),
\end{equation}
where the random variables $(\rv{X},\rv{Y})$ are distributed according to $p_\rv{X}\cdot V^{(n)}_{\rv{Y}|\rv{X}}$ for the conditional entropy $H(\rv{Y}|\rv{X})$ in the above expressions.
Recall~\eqref{eq:size:conditional:type} (also see Lemma~\ref{lem:conditional:type:size}~\cite[Lemma~2.5]{csiszar2011information}), and we can further write 
\begin{align}
    &\begin{aligned}
    1-\epsilon^\fnc{NS}(\floor{e^{nr}},W_{\rv{Y}|\rv{X}}^{\tensor n})\geq\frac{1}{(n+1)^{\size*{\set{X}}-1}} \cdot
    \multiadjustlimits{
        \inf_{p_\rv{X}\in\set{P}_n(\set{X})},
        \sup_{q_\rv{Y}\in\set{P}(\set{Y})},
        \sup_{V_{\rv{Y}|\rv{X}}^{(n)}\in\set{P}_n(\set{Y}|\set{X})}}\ 
    \frac{1}{(n+1)^{\size*{\set{X}}\cdot\size*{\set{Y}}}}\cdot \ldots\hspace{41pt}\\
    \min\Bigg\{ \exp\left(-n\cdot D\infdiv*{p_\rv{X}\cdot V^{(n)}_{\rv{Y}|\rv{X}}}{p_\rv{X}\cdot W_{\rv{Y}|\rv{X}}}\right),\, \floor{e^{nr}}\cdot \exp\left(-n\cdot D\infdiv*{p_\rv{X}\cdot V^{(n)}_{\rv{Y}|\rv{X}}}{p_\rv{X}\times q_\rv{Y}}\right) \Bigg\}
    \end{aligned}\\
    &\begin{aligned}
    \phantom{1-\epsilon^\fnc{NS}(\floor{e^{nr}},W_{\rv{Y}|\rv{X}}^{\tensor n})}\ge \frac{1}{(n+1)^{\size*{\set{X}}-1}} \cdot
    \adjustlimits \inf_{p_\rv{X}\in\set{P}_n(\set{X})} \sup_{V_{\rv{Y}|\rv{X}}^{(n)}\in\set{P}_n(\set{Y}|\set{X})}\
    \frac{1}{(n+1)^{\size*{\set{X}}\cdot\size*{\set{Y}}}}\cdot \ldots\hspace{77pt}\\
    \min\Bigg\{ \exp\left(-n\cdot D\infdiv*{p_\rv{X}\cdot V^{(n)}_{\rv{Y}|\rv{X}}}{p_\rv{X}\cdot W_{\rv{Y}|\rv{X}}}\right),\, \floor{e^{nr}}\cdot \exp\left(-n\cdot D\infdiv*{p_\rv{X}\cdot V^{(n)}_{\rv{Y}|\rv{X}}}{p_\rv{X}\times p^{(n)}_\rv{Y}}\right) \Bigg\}
    \end{aligned}\end{align}
where we choose  $q_{\rv{Y}}=p_\rv{Y}^{(n)}(\cdot) \defeq \sum_{x\in\set{X}} p_\rv{X}(x) \cdot V_{\rv{Y}|\rv{X}}^{(n)}(\cdot|x)$.
Following a similar argument in Section~\ref{sec:EE-converse}, we can restrict $V_{\rv{Y}|\rv{X}}^{(n)}$ to approximations $V_{\rv{Y}|\rv{X}}^{(p_\rv{X})}$ of $V_{\rv{Y}|\rv{X}}\in\set{P}(\set{Y}|\set{X})$ and $p_\rv{X}$ to approximations $p^{(n)}_\rv{X}$ of $p_\rv{X}\in\set{P}(\set{X})$ (see Lemma~\ref{lem:I:continuity} together with Lemma~\ref{lem:type:approx}, Lemma~\ref{lem:conditional:type:approx}, Lemma~\ref{lem:D:continuity}, and~\eqref{eq:approx:V}--\eqref{EE:converse:temp:2}), \ie,
\begin{align}\label{eq:SCE:method:of:types:end}\begin{aligned}
    &1-\epsilon^\fnc{NS}(\floor{e^{nr}},W_{\rv{Y}|\rv{X}}^{\tensor n})\geq\frac{1}{(n+1)^{\size*{\set{X}}\cdot\size*{\set{Y}}+\size*{\set{X}}-1}} \cdot
    \adjustlimits\inf_{p_\rv{X}\in\set{P}(\set{X})}\sup_{V_{\rv{Y}|\rv{X}}\in\set{P}(\set{Y}|\set{X})}\
    \ldots\\
    &\min\Bigg\{\! A_n\!\cdot\! \exp\left(-n\cdot D\infdiv*{p_\rv{X}\cdot V_{\rv{Y}|\rv{X}}}{p_\rv{X}\cdot W_{\rv{Y}|\rv{X}}}\right),\, \floor{e^{nr}}\!\cdot\! \tilde{B}_n^{-1}\!\cdot\!\exp\left(-n\cdot D\infdiv*{p_\rv{X}\cdot V_{\rv{Y}|\rv{X}}}{p_\rv{X}\times p_\rv{Y}}\right) \!\Bigg\}
\end{aligned}\end{align}
for all $V_{\rv{Y}|\rv{X}}\in\set{P}(\set{Y}|\set{X})$, where $A_n$ is the same one as in~\eqref{EE:converse:temp:1} whereas $\tilde{B}_n$ turns out to be as follows
\begin{equation}
    \tilde{B}_n \defeq n^{3\cdot \size*{\set{X}}\cdot\size*{\set{Y}}}\cdot \size*{\set{X}}^{-2\cdot\size*{\set{X}}\cdot\size*{\set{Y}}}\cdot \size*{\set{Y}}^{\size*{\set{X}}}.	
\end{equation}

Notice that $\min\{x,y\} = \min_{\alpha\in[0,1]} x^{\alpha}\cdot y^{1-\alpha} $ for any positive real numbers $x$ and $y$.
Thus, we can rewrite~\eqref{eq:SCE:method:of:types:end} as
\begin{align}
\begin{aligned}
     1-\epsilon^\fnc{NS}(\floor{e^{nr}},W_{\rv{Y}|\rv{X}}^{\tensor n})\geq
      \frac{1}{(n+1)^{\size*{\set{X}}\cdot\size*{\set{Y}}+\size*{\set{X}}-1}} \cdot \multiadjustlimits{
    \inf_{p_\rv{X}\in\set{P}(\set{X})}, 
    \sup_{V_{\rv{Y}|\rv{X}}\in\set{P}(\set{Y}|\set{X})}, 
    \inf_{\alpha\in[0,1]}}\ldots \hspace{10pt} \\
    \ 
     A_n^{\alpha}\cdot \exp\left(-n\cdot \alpha \cdot D\infdiv*{p_\rv{X}\cdot V_{\rv{Y}|\rv{X}}}{p_\rv{X}\cdot W_{\rv{Y}|\rv{X}}}\right) \cdot \left(\floor{e^{nr}}\right)^{1-\alpha}\cdot \ldots\\
     \tilde{B}_n^{\alpha-1}\cdot\exp\left(-n\cdot(1-\alpha)\cdot D\infdiv*{p_\rv{X}\cdot V_{\rv{Y}|\rv{X}}}{p_\rv{X}\times p_\rv{Y}}\right),
\end{aligned}\end{align}
or equivalently, 
\begin{equation}\label{eq:SCE-achievability:temp}\begin{aligned}
    &\hspace{14pt}-\frac{1}{n}\log{\left(1-\epsilon^\fnc{NS}(\floor{e^{nr}},W_{\rv{Y}|\rv{X}}^{\tensor n})\right)}	
    -\frac{\size*{\set{X}}\cdot\size*{\set{Y}}+\size*{\set{X}}-1}{n}\log{(n+1)}
    -\frac{\log{\tilde{B}_n}}{n} \\
    &\begin{aligned}\leq \multiadjustlimits{
    \sup_{p_\rv{X}\in\set{P}(\set{X})}, 
    \inf_{V_{\rv{Y}|\rv{X}}\in\set{P}(\set{Y}|\set{X})}, 
    \sup_{\alpha\in[0,1]}}\ 
    - \alpha\cdot\frac{\log{(A_n\tilde{B}_n)}}{n} + (\alpha-1)\cdot r_n + \ldots \hspace{129pt}\\
    \underbrace{\alpha \cdot D\infdiv*{p_\rv{X}\cdot V_{\rv{Y}|\rv{X}}}{p_\rv{X}\cdot W_{\rv{Y}|\rv{X}}}
    +(1-\alpha)\cdot D\infdiv*{p_\rv{X}\cdot V_{\rv{Y}|\rv{X}}}{p_\rv{X}\times p_\rv{Y}}}_{= I(\rv{X};\rv{Y}) - \alpha\cdot H(\rv{Y}) - \alpha\cdot\sum_{x,y} p_\rv{X}(x)\cdot V_{\rv{Y}|\rv{X}}(y|x) \cdot \log{W_{\rv{Y}|\rv{X}}(y|x)} \text{ with }(\rv{X},\rv{Y})\sim p_\rv{X}\cdot V_{\rv{Y}|\rv{X}} } \end{aligned}
\end{aligned}\end{equation}
where $r_n \defeq \frac{1}{n}\log{\floor{e^{nr}}}$ has already been defined in Proposition~\ref{prop:EE-Ach}.
Note that the expression to be optimized on the RHS of~\eqref{eq:SCE-achievability:temp} is linear in $\alpha$ and convex in $V_{\rv{Y}|\rv{X}}$ and the sets $[0,1]$ and $\set{P}(\set{Y}|\set{X})$ are convex compact.
We apply Sion's minimax theorem~\cite{sion1958general} to swap the supremum over $\alpha$ outside, and get 
\begin{equation}\label{eq:SCE-achievability:temp:2}\begin{aligned}
    &\hspace{14pt}-\frac{1}{n}\log{\left(1-\epsilon^\fnc{NS}(\floor{e^{nr}},W_{\rv{Y}|\rv{X}}^{\tensor n})\right)}	
    -\frac{\size*{\set{X}}\cdot\size*{\set{Y}}+\size*{\set{X}}-1}{n}\log{(n+1)}
    -\frac{\log{\tilde{B}_n}}{n} \\
    &\begin{aligned}\leq \sup_{\alpha\in[0,1]} - \alpha\cdot\frac{\log{(A_n\tilde{B}_n)}}{n} + (\alpha-1)\cdot r_n  + \ldots \hspace{223pt}\\
    \multiadjustlimits{
    \sup_{p_\rv{X}\in\set{P}(\set{X})}, 
    \inf_{V_{\rv{Y}|\rv{X}}\in\set{P}(\set{Y}|\set{X})},
    \inf_{q_\rv{Y}\in\set{P}(\set{Y})}}\ 
    \alpha \cdot D\infdiv*{p_\rv{X}\cdot V_{\rv{Y}|\rv{X}}}{p_\rv{X}\cdot W_{\rv{Y}|\rv{X}}}
    +(1-\alpha)\cdot D\infdiv*{p_\rv{X}\cdot V_{\rv{Y}|\rv{X}}}{p_\rv{X}\times q_\rv{Y}} \end{aligned}
\end{aligned}\end{equation}
where we have also used the fact that $\inf_{q_\rv{Y}\in\set{P}(\set{Y})} D\infdiv*{p_\rv{X}\cdot V_{\rv{Y}|\rv{X}}}{p_\rv{X}\times q_\rv{Y}} = D\infdiv*{p_\rv{X}\cdot V_{\rv{Y}|\rv{X}}}{p_\rv{X}\times p_\rv{Y}}$.
By the variational formulation of the R\'enyi divergence~\cite[Theorem 30]{vanErven2014Jun}, we have for $\alpha\in [0,1]$
\begin{equation}\label{eq:variational-Renyi}\begin{aligned}
    \inf_{V_{\rv{Y}|\rv{X}}\in\set{P}(\set{Y}|\set{X})}
    \alpha \cdot D\infdiv*{p_\rv{X}\cdot V_{\rv{Y}|\rv{X}}}{p_\rv{X}\cdot W_{\rv{Y}|\rv{X}}}
    +(1-\alpha)\cdot D\infdiv*{p_\rv{X}\cdot V_{\rv{Y}|\rv{X}}}{p_\rv{X}\times q_\rv{Y}} \hspace{97pt}\\
    = (1-\alpha) \cdot \mathbb{E}_{x\sim p_\rv{X}}\left[D_\alpha\infdiv*{ W_{\rv{Y}|\rv{X}}(\cdot|x)}{q_\rv{Y}}\right].
\end{aligned}\end{equation}
Combine~\eqref{eq:variational-Renyi} with~\eqref{eq:SCE-achievability:temp:2}, we have 
\begin{align}
    &\hspace{14pt}-\frac{1}{n}\log{\left(1-\epsilon^\fnc{NS}(\floor{e^{nr}},W_{\rv{Y}|\rv{X}}^{\tensor n})\right)}\nonumber \\
    &\leq \multiadjustlimits{
        \sup_{\alpha\in[0,1]},
        \sup_{p_\rv{X}\in\set{P}(\set{X})},
        \inf_{q_\rv{Y}\in\set{P}(\set{Y})}}
    (1-\alpha)\cdot \left(\mathbb{E}_{x\sim p_\rv{X}}\left[D_\alpha\infdiv*{ W_{\rv{Y}|\rv{X}}(\cdot|x)}{q_\rv{Y}}\right] - r_n\right) 
    -\alpha\cdot\tilde{g}_n + \tilde{f}_n \\
    &= \multiadjustlimits{
        \sup_{\alpha\in[0,1]},
        \sup_{p_\rv{X}\in\set{P}(\set{X})},
        \inf_{q_\rv{Y}\in\set{P}(\set{Y})}}
    (1-\alpha)\cdot \left(D_{\alpha}\infdiv*{p_\rv{X}\cdot W_{\rv{Y}|\rv{X}}}{p_\rv{X}\times q_\rv{Y}}- r_n\right) 
    -\alpha\cdot\tilde{g}_n + \tilde{f}_n
\end{align}
 where we use \cite[Proposition 1]{csiszar1995generalized} in the last equality  and
\begin{align*}
 \tilde{f}_n &\defeq \frac{\size*{\set{X}}\cdot\size*{\set{Y}}+\size*{\set{X}}-1}{n}\log{(n+1)} + \frac{\log{\tilde{B}_n}}{n}\\
&= \frac{\left(\size*{\set{X}}\cdot\size*{\set{Y}}+\size*{\set{X}}-1\right)\cdot\log{(n+1)}}{n} + \frac{3\cdot\size*{\set{X}}\cdot\size*{\set{Y}}\cdot\log{n}}{n} + \frac{-2\cdot\size*{\set{X}}\cdot\size*{\set{Y}}\cdot\log{\size*{\set{X}}} + \size*{\set{X}}\cdot\log{\size*{\set{Y}}}}{n}\\
&= O\left(\frac{\log{n}}{n}\right), \\
 \tilde{g}_n &\defeq \frac{1}{n}\log{\left(A_n \tilde{B}_n\right)} 
= \frac{2\cdot\size*{\set{X}}\cdot\size*{\set{Y}}\cdot\log{(W_{\min}\cdot n)}}{n} + \frac{-2\cdot\size*{\set{X}}\cdot\size*{\set{Y}}\cdot\log{\size*{\set{X}}} + \size*{\set{X}}\cdot\log{\size*{\set{Y}}}}{n} \\
&= O\left(\frac{\log{n}}{n}\right).\qedhere
\end{align*}
\end{proof}
From Propositions \ref{prop:SCE:Converse} and \ref{prop:SCE:achievability}, we deduce strong converse exponent for non-signaling channel simulation. 
\begin{theorem}\label{thm:NS:SCE}
Let $W_{\rv{Y}|\rv{X}}\in\set{P}(\set{Y}|\set{X})$ be a channel.
For all $r> 0$, it holds that 
\begin{equation}
    \lim_{n\to \infty} -\frac{1}{n} \log\left(1-\epsilon^\fnc{NS}(\floor{e^{nr}},W_{\rv{Y}|\rv{X}}^{\tensor n})\right) =
    \sup_{\alpha\in[0,1]} \; (1-\alpha)\cdot\left( -r + \adjustlimits\sup_{p_\rv{X}\in\set{P}(\set{X})} \inf_{q_\rv{Y}\in\set{P}(\set{Y})} D_{\alpha}\infdiv*{p_\rv{X}\cdot W_{\rv{Y}|\rv{X}}}{p_\rv{X}\times q_\rv{Y}} \right).
\end{equation}
\end{theorem}
\section{Exponents in the Shared Randomness-Assisted Scenario}\label{sec:SR}
In this section, we would like to link the exponents of channel simulation in the shared randomness-assisted scenario to those in the non-signaling scenario.
In particular, in~\cite{berta2024optimality}, the authors have shown the following result connecting the optimal deviations for channel simulation under the two scenarios, \ie,
\begin{theorem}[{\cite[Corollary~1.1]{berta2024optimality}}, rephrased] \label{prop:rounding}
Let $W_{\rv{Y}|\rv{X}}\in\set{P}(\set{Y}|\set{X})$ be a channel.
For any $M, M' \geq 1$, it holds that
\begin{equation}
    \epsilon^\fnc{NS}(M',W_{\rv{Y}|\rv{X}}) \leq
    \epsilon^\fnc{SR}(M',W_{\rv{Y}|\rv{X}}) \leq
    \left(1-\left(1-\frac{1}{M}\right)^{M'}\right)\cdot \epsilon^\fnc{NS}(M,W_{\rv{Y}|\rv{X}}) + \left(1-\frac{1}{M}\right)^{M'}
\end{equation}
where $\epsilon^\fnc{SR}(M',W_{\rv{Y}|\rv{X}})$ denotes the minimal attainable distortion (measured in TVD) of shared randomness-assisted simulation codes for $W_{\rv{Y}|\rv{X}}$ with alphabet size at most $M'$.
\end{theorem}

\subsection{Error Exponent for Shared Randomness-Assisted Channel Simulation}
We use Theorem~\ref{prop:rounding} to establish the following relationships between the error exponents for channel simulation under the two scenarios.
\begin{lemma}\label{lem:SR-NS-EE}
Let $W_{\rv{Y}|\rv{X}}\in\set{P}(\set{Y}|\set{X})$ be a channel.
For any $r>0$ and $\delta>0$, it holds that
\begin{align}
\label{eq:rounding:1}
\limsup_{n\to\infty} -\frac{1}{n} \log{\epsilon^\fnc{SR}(\floor{e^{nr}}, W_{\rv{Y}|\rv{X}}^{\tensor n})} 
&\leq \lim_{n\to\infty} -\frac{1}{n} \log{\epsilon^\fnc{NS}(\floor{e^{nr}}, W_{\rv{Y}|\rv{X}}^{\tensor n})},\\
\label{eq:rounding:2}
\liminf_{n\to\infty} -\frac{1}{n} \log{\epsilon^\fnc{SR}(\floor{e^{n\cdot (r+\delta)}}, W_{\rv{Y}|\rv{X}}^{\tensor n})} 
&\geq \lim_{n\to\infty} -\frac{1}{n} \log{\epsilon^\fnc{NS}(\floor{e^{nr}}, W_{\rv{Y}|\rv{X}}^{\tensor n})}.
\end{align}
\end{lemma}
\begin{proof}
Eq.~\eqref{eq:rounding:1} is an immediate implication of the first inequality of Theorem~\ref{prop:rounding}.
To show~\eqref{eq:rounding:2}, we apply the second inequality in Theorem~\ref{prop:rounding}, and get
\begin{align}
    \epsilon^\fnc{SR}(\floor{e^{n\cdot(r+\delta)}}, W_{\rv{Y}|\rv{X}}^{\tensor n})
    & \leq \left(1-\left(1-\frac{1}{\floor{e^{nr}}}\right)^{\floor{e^{n\cdot(r+\delta)}}}\right)\cdot \epsilon^\fnc{NS}(\floor{e^{nr}},W_{\rv{Y}|\rv{X}}^{\tensor n}) + \left(1-\frac{1}{\floor{e^{nr}}}\right)^{\floor{e^{n\cdot(r+\delta)}}} \\
    &\leq \epsilon^\fnc{NS}(\floor{e^{nr}},W_{\rv{Y}|\rv{X}}^{\tensor n}) + \exp\left(-e^{n\cdot\delta}+e^{-n\cdot r}\right),\label{eq:epsSR:epsNS:delta}
\end{align}
where we have used the fact that, for $n\ge \frac{1}{r}$,  
\begin{align}
\log\left\{\left(1-\frac{1}{\floor{e^{nr}}}\right)^{\floor{e^{n\cdot(r+\delta)}}}\right\}
&= \floor{e^{n\cdot(r+\delta)}} \cdot \log{\left(1-\frac{1}{\floor{e^{nr}}}\right)}
\leq \floor{e^{n\cdot(r+\delta)}} \cdot \left(-\frac{1}{\floor{e^{nr}}}\right)\\
&\leq -\frac{e^{n\cdot(r+\delta)}-1}{e^{nr}} = -e^{n\cdot\delta} + e^{-n\cdot r}.
\end{align}
We further claim that for any positive sequence $\{a_n\}_{n\in\mathbb{N}}$ such that $\{\frac{1}{n}\log{a_n}\}_{n\in\mathbb{N}}$ is convergent, it holds that 
\begin{equation}\label{eq:exp:growth:pertubation}
\lim_{n\to\infty} \frac{1}{n} \log{a_n} = \lim_{n\to\infty} \frac{1}{n} \log{\left(a_n + \exp\left(-e^{n\cdot\delta} + e^{-n\cdot r}\right)\right)}
\end{equation}
for any $r, \delta>0$.
We defer the proof of~\eqref{eq:exp:growth:pertubation} to Appendix~\ref{app:exp:growth:pertubation}.
Finally,~\eqref{eq:rounding:2} is the result of combining~\eqref{eq:epsSR:epsNS:delta} and~\eqref{eq:exp:growth:pertubation}.
\end{proof}
Lemma~\ref{lem:SR-NS-EE} enables the following theorem.
\begin{theorem}\label{thm:SR:EE}
Let $W_{\rv{Y}|\rv{X}}\in\set{P}(\set{Y}|\set{X})$ be a channel.
For all $r> 0$, it holds that 
\begin{equation}\label{eq:EE-SR}
    \lim_{n\to \infty} -\frac{1}{n} \log\epsilon^\fnc{SR}(\floor{e^{nr}},W_{\rv{Y}|\rv{X}}^{\tensor n}) =
    \sup_{\alpha\geq 0} \; \alpha\cdot\left( r - \adjustlimits\sup_{p_\rv{X}\in\set{P}(\set{X})} \inf_{q_\rv{Y}\in\set{P}(\set{Y})} D_{\alpha+1}\infdiv*{p_\rv{X}\cdot W_{\rv{Y}|\rv{X}}}{p_\rv{X}\times q_\rv{Y}} \right).
\end{equation}
\end{theorem}
\begin{proof}
By~\eqref{eq:rounding:2} and Theorem~\ref{thm:NS:EE}, we have
\begin{align}
&\hspace{14pt}\liminf_{n\to\infty} -\frac{1}{n} \log{\epsilon^\fnc{SR}(\floor{e^{nr}}, W_{\rv{Y}|\rv{X}}^{\tensor n})} \nonumber \\
&\ge \sup_{\delta\in(0,r)} \lim_{n\to\infty} -\frac{1}{n} \log{\epsilon^\fnc{NS}(\floor{e^{n\cdot(r-\delta)}}, W_{\rv{Y}|\rv{X}}^{\tensor n})}
\\
&= \sup_{\delta\in(0,r)} \sup_{\alpha\geq 0} \; \alpha\cdot\left( (r-\delta) - \adjustlimits\sup_{p_\rv{X}\in\set{P}(\set{X})} \inf_{q_\rv{Y}\in\set{P}(\set{Y})} D_{\alpha+1}\infdiv*{p_\rv{X}\cdot W_{\rv{Y}|\rv{X}}}{p_\rv{X}\times q_\rv{Y}} \right)\\
&=\sup_{\alpha\geq 0}\; \alpha\cdot \sup_{\delta\in(0,r)} \left( (r-\delta) - \adjustlimits\sup_{p_\rv{X}\in\set{P}(\set{X})} \inf_{q_\rv{Y}\in\set{P}(\set{Y})} D_{\alpha+1}\infdiv*{p_\rv{X}\cdot W_{\rv{Y}|\rv{X}}}{p_\rv{X}\times q_\rv{Y}} \right)\\
&= \sup_{\alpha\geq 0}\; \alpha\cdot \left( r - \adjustlimits\sup_{p_\rv{X}\in\set{P}(\set{X})} \inf_{q_\rv{Y}\in\set{P}(\set{Y})} D_{\alpha+1}\infdiv*{p_\rv{X}\cdot W_{\rv{Y}|\rv{X}}}{p_\rv{X}\times q_\rv{Y}} \right).
\end{align}
On the other hand, by~\eqref{eq:rounding:1} and Theorem~\ref{thm:NS:EE} we have 
\begin{align}
\limsup_{n\to\infty} -\frac{1}{n} \log{\epsilon^\fnc{SR}(\floor{e^{nr}}, W_{\rv{Y}|\rv{X}}^{\tensor n})} 
&\leq \lim_{n\to\infty} -\frac{1}{n} \log{\epsilon^\fnc{NS}(\floor{e^{n r}}, W_{\rv{Y}|\rv{X}}^{\tensor n})}
\\&= \sup_{\alpha\geq 0} \; \alpha\cdot\left( r - \adjustlimits\sup_{p_\rv{X}\in\set{P}(\set{X})} \inf_{q_\rv{Y}\in\set{P}(\set{Y})} D_{\alpha+1}\infdiv*{p_\rv{X}\cdot W_{\rv{Y}|\rv{X}}}{p_\rv{X}\times q_\rv{Y}} \right),
\end{align}
which finishes the proof.
\end{proof}

\subsection{Strong Converse Exponent for Shared Randomness-Assisted Channel Simulation}
By taking $M=M'$ in Theorem~\ref{prop:rounding}, we have the following relationships between the deviations for channel simulation under the non-signaling and the shared randomness-assisted scenarios.
\begin{lemma}\label{lem:rounding:success}
Let $W_{\rv{Y}|\rv{X}}\in\set{P}(\set{Y}|\set{X})$ be a channel.
For any $M\geq 1$, it holds that
\begin{equation}\label{eq:rounding:succes}
    1-\epsilon^\fnc{NS}(M,W_{\rv{Y}|\rv{X}}) \geq
    1-\epsilon^\fnc{SR}(M,W_{\rv{Y}|\rv{X}}) \geq
    \left(1-\frac{1}{e}\right)\cdot \left(1-\epsilon^\fnc{NS}(M,W_{\rv{Y}|\rv{X}})\right).
\end{equation}
\end{lemma}
\begin{proof}
Let $M=M'$ in Theorem~\ref{prop:rounding}, we have 
\begin{equation}
    1-\epsilon^\fnc{NS}(M,W_{\rv{Y}|\rv{X}}) \geq
    1-\epsilon^\fnc{SR}(M,W_{\rv{Y}|\rv{X}}) \geq
    \left(1-\left(1-\frac{1}{M}\right)^{M}\right)\cdot \left(1-\epsilon^\fnc{NS}(M,W_{\rv{Y}|\rv{X}})\right).
\end{equation}
This results in~\eqref{eq:rounding:succes} by noting that $(1-\frac{1}{M})^{M}\leq \frac{1}{e}$ for all $M\geq 1$.
\end{proof}
\begin{theorem}\label{thm:SR:SCE}
Let $W_{\rv{Y}|\rv{X}}\in\set{P}(\set{Y}|\set{X})$ be a channel.
For all $r> 0$, it holds that 
\begin{equation}
    \lim_{n\to \infty} -\frac{1}{n} \log\left(1-\epsilon^\fnc{SR}(\floor{e^{nr}},W_{\rv{Y}|\rv{X}}^{\tensor n})\right) =
    \sup_{\alpha\in[0,1]} \; (1-\alpha)\cdot\left( -r + \adjustlimits\sup_{p_\rv{X}\in\set{P}(\set{X})} \inf_{q_\rv{Y}\in\set{P}(\set{Y})} D_{\alpha}\infdiv*{p_\rv{X}\cdot W_{\rv{Y}|\rv{X}}}{p_\rv{X}\times q_\rv{Y}} \right).
\end{equation}
\end{theorem}
\begin{proof}
By Theorem \ref{thm:NS:SCE}, it suffices to show that 
\begin{equation}
\lim_{n\to \infty} -\frac{1}{n} \log\left(1-\epsilon^\fnc{SR}(\floor{e^{nr}},W_{\rv{Y}|\rv{X}}^{\tensor n})\right)
= \lim_{n\to \infty} -\frac{1}{n} \log\left(1-\epsilon^\fnc{NS}(\floor{e^{nr}},W_{\rv{Y}|\rv{X}}^{\tensor n})\right),
\end{equation}
which can be directly shown via Lemma~\ref{lem:rounding:success}, \ie, 
\begin{align}
\liminf_{n\to \infty} -\frac{1}{n} \log\left(1-\epsilon^\fnc{SR}(\floor{e^{nr}},W_{\rv{Y}|\rv{X}}^{\tensor n})\right)
& \geq \lim_{n\to \infty} -\frac{1}{n} \log\left(1-\epsilon^\fnc{NS}(\floor{e^{nr}},W_{\rv{Y}|\rv{X}}^{\tensor n})\right),\\
\limsup_{n\to \infty} -\frac{1}{n} \log\left(1-\epsilon^\fnc{SR}(\floor{e^{nr}},W_{\rv{Y}|\rv{X}}^{\tensor n})\right)
& \leq \lim_{n\to \infty} -\frac{1}{n} \log\left(\left(1-\frac{1}{e}\right)\cdot \left(1-\epsilon^\fnc{NS}(\floor{e^{nr}},W_{\rv{Y}|\rv{X}})\right)\right)\\
& = \lim_{n\to \infty} -\frac{1}{n} \log\left(1-\epsilon^\fnc{NS}(\floor{e^{nr}},W_{\rv{Y}|\rv{X}}^{\tensor n})\right).\qedhere
\end{align}
\end{proof}
\section{Conclusion}\label{sec:conclusion}
In this study, we derived the error exponent and the strong converse exponent for channel simulation in the non-signaling and the shared randomness-assisted settings.
We expressed the exponents in terms of the R\'enyi divergences.
It should be noted that our expressions of exponents hold for all rates $r>0$, \ie, there is no critical rate, despite the fact that the expressions only become nontrivial when $r$ is above or below certain thresholds.
Our finding is quite surprising as compared to the error exponent of channel coding, in which the critical rate is involved~\cite{gallager1968information}.
The method of types has proven to be a useful tool in analyzing the relevant terms in the non-signaling regimes.
Along with earlier findings on small- and moderate-deviation analysis, this research offers a comprehensive study of finite-blocklength analysis of the channel simulation task.
For future studies, it would be intriguing to investigate whether the methods and insights presented here can be generalized to address the unresolved issues in channel interconversion, especially, the corresponding error exponent and the strong converse exponent.

\section*{Acknowledgments}

AO, MC, and MB acknowledge funding by the European Research Council (ERC Grant Agreement No. 948139). MB acknowledges support from the Excellence Cluster - Matter and Light for Quantum Computing (ML4Q). 
HC is supported by NSTC 113-2119-M-001-009, No.~NSTC 113-2628-E-002-029, No.~NTU-113V1904-5, No.~NTU-CC-113L891605, and No.~NTU-113L900702.
We thank Marco Tomamichel for discussions.

\printbibliography
\appendix
\section{Classical Post-Selection Lemma with Integral Representation}
\label{app:proof:lem:perm-inv-de-finetti}
\begin{lemma*}[Lemma \ref{lem:perm-inv-de-finetti} restated]
Let $\set{A}$ be a finite set.
There exists a probability measure $\nu$ on the set of all pmfs on $\set{A}$, \ie, $\set{P}(\set{A})$, such that for all positive integers $n$  
\begin{equation}
    p_{\rvs{A}_1^n}\leq \binom{n+\size*{\set{A}}-1}{n} \cdot \int \diff{\nu}(q_\rv{A}) q_\rv{A}^{\tensor n}
\end{equation}
for every permutation invariant pmf $p_{\rvs{A}_1^n}\in\set{P}(\set{A}^n)$.
\end{lemma*}
\begin{proof}
We use some elements from quantum information theory to construct the desired measure $\nu$.

Let $\hilbert$ be a Hilbert space with a basis $\{\ket{a}\}_{a\in \set{A}}$ indexed by the set $\set{A}$, and let $\sys{A}_1,\ldots,\sys{A}_n$ be $n$ quantum systems, each with the same state space $\hilbert$.
For each type $p\in \set{P}_n(\set{A})$ (see Definition~\ref{def:type} for notations), define $\ket{\psi_p}$ to be the following pure state on the joint system $\syss{A}_1^n$
\begin{align}\label{eq:def-psi}
    \ket{\psi_p} \defeq \frac{1}{\sqrt{\size*{\set{T}_{n}(p)}}}\sum_{\mathbf{a}_1^n\in \set{T}_{n}(p)} \ket{a_1}\otimes \ket{a_2}\otimes \cdots \otimes \ket{a_n}.
\end{align}
Observe that $\ket{\psi_p}$ is invariant under any permutations of $\sys{A}_1, \sys{A}_2, \ldots, \sys{A}_n$, \ie, $\ket{\psi_p}$ lies inside the symmetric subspace $\operatorname{Sym}^n(\mathcal{H})$.
Since the pmf $p_{\rvs{A}_1^n}$ is permutation invariant, we can write 
\begin{equation}
    p_{\rvs{A}_1^n}(\mathbf{a}_1^n) = \sum_{p\in \set{P}_n(\set{A})} \alpha_p \cdot \underbrace{\frac{1}{|\set{T}_{n}(p)|}\bid\{\mathbf{a}_1^n\in \set{T}_{n}(p)\}}_{\defas \uniform_p(\mathbf{a}_1^n)}
    \qquad \forall \mathbf{a}_1^n\in \set{X}^n
\end{equation}
for some pmf $\{\alpha_p\}_{p\in \set{P}_n(\set{A})}$ on $\set{P}_n(\set{A})$.
Correspondingly, we define $\rho_{\syss{A}_1^n}$ to be the following density operator on $\syss{A}_1^n$
\begin{align}\label{eq:rho-alpha}
    \rho_{\syss{A}_1^n} =\sum_{p\in \set{P}_n(\set{A})} \alpha_p \cdot \proj{\psi_p}.
\end{align}
By construction, $\rho_{\syss{A}_1^n}$ is supported on $\operatorname{Sym}^n(\mathcal{H})$.
Thus, by~\cite{christandl2012reliable}, it holds that 
\begin{equation}\label{eq:perm-inv-de-finetti:quantum}
    \rho_{\syss{A}_1^n} \mle \binom{n+\size*{\set{A}}-1}{n} \cdot \int \diff{\ket{\phi}}\; \proj{\phi}^{\otimes n}
\end{equation}
where $\diff{\ket{\phi}}$ is the Haar measure on the Hilbert space $\mathcal{H}$.

Now, we apply the following quantum channel (CPTP map) 
\begin{equation}
\mathcal{M} : \rho \mapsto \sum_{\mathbf{a}_1^n \in \set{A}^n} \bra{\mathbf{a}_1^n}\rho\ket{\mathbf{a}_1^n}  \, \proj{\mathbf{a}_1^n}
\end{equation}
on both sides of the inequality~\eqref{eq:perm-inv-de-finetti:quantum}.
By direct calculation, we have from \eqref{eq:rho-alpha} and \eqref{eq:def-psi}
\begin{align}
    \mathcal{M}(\text{LHS of~\eqref{eq:perm-inv-de-finetti:quantum}})
    &=\sum_{p\in \set{P}_n(\set{A})} \alpha_p \cdot \sum_{\mathbf{a}_1^n \in \set{A}^n}  \abs*{\spr{\mathbf{a}_1^n}{\psi_p}}^2 \cdot \proj{\mathbf{a}_1^n}\\
    &= \sum_{p\in \set{P}_n(\set{A})} \alpha_p \cdot \sum_{\mathbf{a}_1^n \in \set{T}_n(p)}  \frac{1}{\size*{\set{T}_n(p)}} \cdot \proj{\mathbf{a}_1^n}\\
    &= \sum_{p\in \set{P}_n(\set{A})} \alpha_p \cdot \uniform_p = p_{\rvs{A}_1^n},\\
    \mathcal{M}(\text{RHS of~\eqref{eq:perm-inv-de-finetti:quantum}})
    &= \binom{n+\size*{\set{A}}-1}{n} \cdot \int \diff{\ket{\phi}}\; \sum_{\mathbf{a}_1^n \in \set{A}^n}  \abs*{\left\langle\mathbf{a}_1^n\middle\vert\phi\right\rangle^{\otimes n}}^2 \cdot \proj{\mathbf{a}_1^n}\\
    &= \binom{n+\size*{\set{A}}-1}{n} \cdot \int \diff{\ket{\phi}}\; \sum_{\mathbf{a}_1^n \in \set{A}^n} \prod_{i=1}^n \abs*{\spr{a_i}{\phi}}^2 \cdot \proj{\mathbf{a}_1^n}\\
    &= \binom{n+\size*{\set{A}}-1}{n} \cdot \int \diff{\ket{\phi}}\; \bigg(\underbrace{\sum_{a \in \set{A}} \abs*{\spr{a}{\phi}}^2 \cdot \proj{a}}_{\defas q_{\ket{\phi}}}\bigg)^{\tensor n},
\end{align}
where we have treated pmfs as diagonal density operators, \ie, $p_{\rvs{A}_1^n} \equiv \sum_{\mathbf{a}_1^n\in\set{A}^n} p_{\rvs{A}_1^n}(\mathbf{a}_1^n)\cdot \proj{\mathbf{a}_1^n}$.
By linearity and complete-positiveness of $\mathcal{M}$, we have
\begin{equation}\label{eq:mixture-iid}
    p_{\rvs{A}_1^n} \leq \binom{n+\size*{\set{A}}-1}{n} \cdot \int \diff{\ket{\phi}}\; q_{\ket{\phi}}^{\tensor n}
\end{equation}
where $q_{\ket{\phi}}: a \mapsto \abs*{\spr{a}{\phi}}^2$ is a pmf on $\set{A}$ for all pure state $\ket{\phi}$.
\end{proof}

\begin{remark}\label{app:rmk:lem:perm-inv-de-finetti}
In the following remark, we show the pmf on the RHS of~\eqref{eq:mixture-iid} equals the universal probability distribution as in~\cite{tomamichel2017operational} exactly, \ie, 
\begin{equation}
    \left(\int \diff{\ket{\phi}}\; q_{\ket{\phi}}^{\tensor n}\right)(\mathbf{a}_1^n) = 
    \int \diff{\ket{\phi}}\; \prod_{i=1}^n q_{\ket{\phi}}(a_i) = 
    \frac{1}{\size*{\set{P}_n(\set{A})}} \cdot \sum_{p\in\set{P}_n(\set{A})} \frac{1}{\size*{\set{T}_n(p)}} \mathbbm{1}\left\{\mathbf{a}_1^n \in \set{T}_n(p)\right\}
    \quad \forall \mathbf{a}_1^n\in\set{A}^n.
\end{equation}
Let $\diff{U}$ be the Haar measure over the set of unitary operators over $\hilbert$ \wrt the Hilbert-Schmidt inner product.
This allows us to replace $\diff{\nu}$ by $\diff{U}$ as follows
\begin{equation}
\int \diff{\ket{\phi}}\; \prod_{i=1}^n q_{\ket{\phi}}(a_i)
= \int \diff{U}\; \prod_{i=1}^n q_{U\ket{0}}(a_i)
= \int \diff{U}\; \prod_{i=1}^n \bra{0} U^\dagger \proj{a_i} U \ket{0}
\end{equation}
where $\ket{0}\in \{\ket{a}\}_{a\in \set{A}}$ is a fixed pure state. 
Using the Weingarten formula  for the unitary group $\mathbb{U}(\size*{\set{A}})$ (e.g., \cite{collins2022weingarten}), we have
\begin{equation}
 \int \diff{U}\; \prod_{i=1}^n \bra{0} U^\dagger \proj{a_i} U \ket{0}= \sum_{\sigma, \tau \in S_n} \mathbbm{1}\left\{\mathbf{a}_1^n = \sigma(\mathbf{a}_1^n)\right\} \cdot \fnc{Wg}(\sigma^{-1}\tau, \size*{\set{A}})
\end{equation}
where $S_n$ is the symmetric group of $n$ elements (permutations), for each permutation $\sigma\in S_n$, we denote $\sigma(\mathbf{a}_1^n) \defeq (a_{\sigma(1)}, \ldots, a_{\sigma(n)})$, and $\fnc{Wg}$ is the Weingarten function. 
Note that (see, \eg,~\cite[Eq.~(2.3)]{collins2014integration})
\begin{equation}
\sum_{\tau\in S_n} \fnc{Wg}(\sigma^{-1}\tau, \size*{\set{A}})
=\sum_{\tau\in S_n} \fnc{Wg}(\tau, \size*{\set{A}})
= \frac{1}{\size*{\set{A}}\cdot (\size*{\set{A}}+1)  \cdots  (\size*{\set{A}}+n-1)}.
\end{equation}
Observe that the permutations of the sequence $\mathbf{a}_1^n$ shall visit all sequences of the same type for equal number of times.
Thus, 
\begin{equation}
\sum_{\sigma\in S_n} \mathbbm{1}\left\{\mathbf{a}_1^n = \sigma(\mathbf{a}_1^n)\right\}
= \frac{\size*{S_n}}{\size*{\set{T}_n(p^{(\mathbf{a}_1^n)}_\rv{A})}}
= \frac{n!}{\size*{\set{T}_n(p^{(\mathbf{a}_1^n)}_\rv{A})}}
\end{equation}
where $p^{(\mathbf{a}_1^n)}_\rv{A}$ denotes the type (or the induced empirical distribution) of the sequence $\mathbf{a}_1^n$.
Combining the above, we have 
\begin{align}
\int \diff{\ket{\phi}}\; \prod_{i=1}^n q_{\ket{\phi}}(a_i)
&= \frac{n!}{\size*{\set{T}_n(p^{(\mathbf{a}_1^n)}_\rv{A})}} \cdot \frac{1}{\size*{\set{A}}\cdot (\size*{\set{A}}+1)  \cdots  (\size*{\set{A}}+n-1)}\\
&= \frac{1}{\size*{\set{T}_n(p^{(\mathbf{a}_1^n)}_\rv{A})}}\cdot \frac{1}{\binom{n+\size*{\set{A}}-1}{n}}.
\end{align}
Further note that 
\begin{equation}
\size*{\set{P}_n(\set{A})} = \binom{n+\size*{\set{A}}-1}{n}.
\end{equation}
Therefore, for all $\mathbf{a}_1^n\in\set{A}^n$
\begin{equation}
\left(\int \diff{\ket{\phi}}\; q_{\ket{\phi}}^{\tensor n}\right)(\mathbf{a}_1^n)
= \frac{1}{\size*{\set{P}_n(\set{A})}}\cdot \frac{1}{\size*{\set{T}_n(p^{(\mathbf{a}_1^n)}_\rv{A})}}
= \frac{1}{\size*{\set{P}_n(\set{A})}}\cdot \sum_{p\in\set{P}_n(\set{A})} \frac{1}{\size*{\set{T}_n(p)}} \cdot \mathbbm{1}\left\{\mathbf{a}_1^n \in \set{T}_n(p)\right\}.
\end{equation}

\end{remark}

\section{Proofs of Lemmas~\ref{lem:conditional:type:approx},~\ref{lem:D:continuity}, and~\ref{lem:I:continuity}}\label{app:continuity:proofs}
\begin{lemma*}[Lemma \ref{lem:conditional:type:approx} restated]
Let $\set{X}$ and $\set{Y}$ be two finite sets, and let $n$ be a positive integer.
Let $V_{\rv{Y}|\rv{X}}\in\set{P}(\set{Y}|\set{X})$ be a conditional pmf.
For each $n$-denominator type on $\set{X}$ $p_\rv{X}\in\set{P}_n(\set{X})$, there exists an approximation $V_{\rv{Y}|\rv{X}}^{(p_\rv{X})}\in\set{P}_n(\set{Y}|\set{X})$ of $V_{\rv{Y}|\rv{X}}$ such that
\begin{equation}\label{eq:conditional:type:approx:app}
    \abs*{p_\rv{X}(x)\cdot V_{\rv{Y}|\rv{X}}^{(p_\rv{X})}(y|x) - p_\rv{X}(x)\cdot V_{\rv{Y}|\rv{X}}(y|x)} \leq \frac{1}{n}
\end{equation}
for all $(x,y)\in\set{X}\times\set{Y}$.
\end{lemma*}
\begin{proof}[Proof of Lemma~\ref{lem:conditional:type:approx}]
Pick a sequence $\mathbf{x}_1^n$ from $\set{T}_n(p_\rv{X})$.
We argue that there exists some output sequence $\mathbf{y}_1^n\in\set{Y}_1^n$ such that
\begin{equation}
    \floor*{n\cdot p_{\rv{X}}\cdot V_{\rv{Y}|\rv{X}}} \leq
    n\cdot p_{\rv{XY}}^{(\mathbf{x}_1^n,\mathbf{y}_1^n)} \leq
    \ceil*{n\cdot p_{\rv{X}}\cdot V_{\rv{Y}|\rv{X}}}
\end{equation}
since $\sum_{x,y} n\cdot p_\rv{X}(x) \cdot V_{\rv{Y}|\rv{X}}(y|x) = n$.
Eq.~\eqref{eq:conditional:type:approx:app} is satisfied by defining $V_{\rv{Y}|\rv{X}}^{(p_\rv{X})}$ to be the conditional type of $(\mathbf{x}_1^n, \mathbf{y}_1^n)$, \ie, $V_{\rv{Y}|\rv{X}}^{(p_\rv{X})} \defeq p_{\rv{Y}|\rv{X}}^{(\mathbf{x}_1^n,\mathbf{y}_1^n)}$.
\end{proof}
\begin{lemma}[Lemma \ref{lem:D:continuity} restated]
Let $\set{A}$ be a finite set, and $\xi\in(0,\frac{1}{e})$
For any two pmfs $p,p'\in\set{P}(\set{A})$ such that $\abs*{p(a)-p'(a)}\leq\xi$ { for all $a\in\set{A}$}, it holds for all $q\in\set{P}(\set{A})$, with $p\ll q$ and $p'\ll q$, that
\begin{equation}
\abs*{D\infdiv*{p}{q} - D\infdiv*{p'}{q}} \leq \xi\cdot \size*{\set{A}} \cdot \log{\frac{1}{q_{\min}}} + \size*{\set{A}}\cdot\xi\cdot\log{\frac{1}{\xi}}
\end{equation}
where $q_{\min}\defeq \min_{a\in\set{A}:\, q(a)>0} q(a)$.
\end{lemma}
\begin{proof}[Proof of Lemma~\ref{lem:D:continuity}]
This is a result using the triangular inequality on absolute values and a property of the function $t\mapsto t\cdot\log{t}$.
\begin{align}
\abs*{D\infdiv*{p}{q} - D\infdiv*{p'}{q}}
&= \abs*{\sum_{a\in\set{A}} p(a)\log{\frac{p(a)}{q(a)}} - p'(a)\log{\frac{p'(a)}{q(a)}}}\\
&\leq \sum_{a\in\set{A}} \abs*{p(a)-p'(a)}\cdot\log{\frac{1}{q(a)}} + \sum_{a\in\set{A}} \abs*{p(a)\log{p(a)} - p'(a)\log{p'(a)}}\\
&\leq \xi\cdot \size*{\set{A}}\cdot \max_{a\in\set{A}: q(a)>0}\log{\frac{1}{q(a)}} + \size*{\set{A}}\cdot \sup_{0\leq t_0\leq t_1\leq 1: t_1-t_0 \leq \xi} \abs*{t_1\log{t_1} - t_0\log{t_0}}\\
&\leq \xi\cdot \size*{\set{A}}\cdot \log{\frac{1}{q_{\min}}} + \size*{\set{A}}\cdot \xi\cdot\log{\frac{1}{\xi}}. \qedhere
\end{align}
\end{proof}
\begin{lemma*}[Lemma \ref{lem:I:continuity} restated]
Let $\set{X}$ and $\set{Y}$ be two finite sets.
Let $p_\rv{X}\in\set{P}(\set{X})$ be a pmf on $\set{X}$, and let $V_{\rv{Y}|\rv{X}}\in\set{P}(\set{Y}|\set{X})$ be a channel from $\set{X}$ to $\set{Y}$.
\begin{enumerate}
\item For any $\tilde{V}_{\rv{Y}|\rv{X}}\in\set{P}(\set{Y}|\set{X})$ such that $\abs{p_\rv{X}(x)\cdot\tilde{V}_{\rv{Y}|\rv{X}}(y|x)-p_\rv{X}(x)\cdot V_{\rv{Y}|\rv{X}}(y|x)}\leq\xi\ \forall (x,y)\in\set{X}\times\set{Y}$ where $0 < \xi\leq \frac{1}{\size*{\set{X}}\cdot e}$, it hold that
    \begin{equation}\label{eq:I:continuity:1:app}
    \abs*{D\infdiv*{p_\rv{X}\cdot V_{\rv{Y}|\rv{X}}}{p_\rv{X}\times p_\rv{Y}} - D\infdiv*{p_\rv{X}\cdot \tilde{V}_{\rv{Y}|\rv{X}}}{p_\rv{X}\times \tilde{p}_\rv{Y}}} \leq \xi\cdot \size*{\set{X}}\cdot\size*{\set{Y}}\cdot \left(\log{\frac{1}{\xi}} + \log{\frac{1}{\xi\cdot \size*{\set{X}}}} \right),
    \end{equation}
    where $p_\rv{Y}$ and $\tilde{p}_\rv{Y}$ are the induced output distributions of the channels $V_{\rv{Y}|\rv{X}}$ and $\tilde{V}_{\rv{Y}|\rv{X}}$ (with the same input distribution $p_\rv{X}$), respectively.
    \item For any $\tilde{p}_\rv{X}\in\set{P}(\set{X})$ such that $\abs{\tilde{p}_\rv{X}(x)-p_\rv{X}(x)}\leq\xi\ \forall x\in\set{X}$ where $0 < \xi\leq \frac{1}{\size*{\set{X}}\cdot e}$, it hold that
    \begin{equation}\label{eq:I:continuity:2:app}
    \abs*{D\infdiv*{p_\rv{X}\cdot V_{\rv{Y}|\rv{X}}}{p_\rv{X}\times p_\rv{Y}} - D\infdiv*{\tilde{p}_\rv{X}\cdot V_{\rv{Y}|\rv{X}}}{\tilde{p}_\rv{X}\times \tilde{p}_\rv{Y}}} \leq \xi\cdot\size*{\set{X}}\cdot\log{\size*{\set{Y}}} + \xi\cdot \size*{\set{X}}\cdot\size*{\set{Y}}\cdot \log{\frac{1}{\xi\cdot \size*{\set{X}}}}
    \end{equation}
    where $p_\rv{Y}$ and $\tilde{p}_\rv{Y}$ are the induced output distributions corresponding to the input distributions $p_\rv{X}$ and $\tilde{p}_\rv{X}$ (with the same channel $V_{\rv{Y}|\rv{X}}$), respectively.
\end{enumerate}
\end{lemma*}
\begin{proof}[Proof of Lemma~\ref{lem:I:continuity}]
The proof utilizes the triangular inequality on absolute values and the property of the function $t\mapsto t\cdot\log{t}$ in a similar way as in the proof of Lemma~\ref{lem:D:continuity}.
To prove~\eqref{eq:I:continuity:1:app}, we have
\begin{align}
    &\hspace{13.3pt}\abs*{D\infdiv*{p_\rv{X}\cdot V_{\rv{Y}|\rv{X}}}{p_\rv{X}\times p_\rv{Y}} - D\infdiv*{p_\rv{X}\cdot \tilde{V}_{\rv{Y}|\rv{X}}}{p_\rv{X}\times \tilde{p}_\rv{Y}}} \nonumber\\
    &= \abs*{
    \sum_{x,y} p_\rv{X}(x)\cdot V_{\rv{Y}|\rv{X}}(y|x) \cdot \log{\frac{V_{\rv{Y}|\rv{X}}(y|x)}{p_\rv{Y}(y)}}
    - \sum_{x,y} p_\rv{X}(x)\cdot \tilde{V}_{\rv{Y}|\rv{X}}(y|x) \cdot \log{\frac{\tilde{V}_{\rv{Y}|\rv{X}}(y|x)}{\tilde{p}_\rv{Y}(y)}}} \\
    &\ \begin{aligned}\leq \sum_{x,y}\abs*{p_\rv{X}(x) \cdot V_{\rv{Y}|\rv{X}}(y|x)\cdot \log{\left(p_\rv{X}(x) \cdot V_{\rv{Y}|\rv{X}}(y|x)\right)} - p_\rv{X}(x) \cdot \tilde{V}_{\rv{Y}|\rv{X}}(y|x)\cdot \log{\left(p_\rv{X}(x) \cdot \tilde{V}_{\rv{Y}|\rv{X}}(y|x)\right)}} \\
    + \sum_{y} \abs*{p_\rv{Y}(y)\log{p_\rv{Y}(y)} - \tilde{p}_\rv{Y}(y)\log{\tilde{p}_\rv{Y}(y)}}\end{aligned}\\
    &\leq \size*{\set{X}}\cdot\size*{\set{Y}}\cdot \xi \cdot \log{\frac{1}{\xi}} + \sum_{y} \abs*{p_\rv{Y}(y)\log{p_\rv{Y}(y)} - \tilde{p}_\rv{Y}(y)\log{\tilde{p}_\rv{Y}(y)}}.
    \label{eq:I:continuity:proof:1}
\end{align}
Note that 
\begin{align}
\abs*{p_\rv{Y}(y) - \tilde{p}_\rv{Y}(y)}
&= \abs*{\sum_{x} p_\rv{X}(x)\cdot V_{\rv{Y}|\rv{X}}(y|x) - p_\rv{X}(x)\cdot \tilde{V}_{\rv{Y}|\rv{X}}(y|x)} \\
&\leq \sum_{x} \abs*{p_\rv{X}(x)\cdot V_{\rv{Y}|\rv{X}}(y|x) - p_\rv{X}(x)\cdot \tilde{V}_{\rv{Y}|\rv{X}}(y|x)} \\
&\leq \size*{\set{X}} \cdot \xi.
\end{align}
Hence, 
\begin{equation}\label{eq:I:continuity:proof:2}
    \sum_{y} \abs*{p_\rv{Y}(y)\log{p_\rv{Y}(y)} - \tilde{p}_\rv{Y}(y)\log{\tilde{p}_\rv{Y}(y)}}
    \leq \size*{\set{Y}} \cdot \size*{\set{X}} \cdot \xi \cdot \log{\frac{1}{\size*{\set{X}}\cdot \xi}}.
\end{equation}
Combining~\eqref{eq:I:continuity:proof:2} with~\eqref{eq:I:continuity:proof:1}, we have 
\begin{equation}\begin{aligned}
&\hspace{14pt}\abs*{D\infdiv*{p_\rv{X}\cdot V_{\rv{Y}|\rv{X}}}{p_\rv{X}\times p_\rv{Y}} - D\infdiv*{p_\rv{X}\cdot \tilde{V}_{\rv{Y}|\rv{X}}}{p_\rv{X}\times \tilde{p}_\rv{Y}}} \\
&\leq \size*{\set{X}}\cdot\size*{\set{Y}}\cdot \xi \cdot \log{\frac{1}{\xi}} + \size*{\set{Y}} \cdot \size*{\set{X}} \cdot \xi \cdot \log{\frac{1}{\size*{\set{X}}\cdot \xi}} 
= \text{RHS of~\eqref{eq:I:continuity:1:app}}.
\end{aligned}\end{equation}

To prove~\eqref{eq:I:continuity:2:app}, we have 
\begin{align}
    &\hspace{13.3pt}\abs*{D\infdiv*{p_\rv{X}\cdot V_{\rv{Y}|\rv{X}}}{p_\rv{X}\times p_\rv{Y}} - D\infdiv*{\tilde{p}_\rv{X}\cdot V_{\rv{Y}|\rv{X}}}{\tilde{p}_\rv{X}\times \tilde{p}_\rv{Y}}}\nonumber\\
    &= \abs*{
    \sum_{x,y} p_\rv{X}(x)\cdot V_{\rv{Y}|\rv{X}}(y|x) \cdot \log{\frac{V_{\rv{Y}|\rv{X}}(y|x)}{p_\rv{Y}(y)}}
    - \sum_{x,y} \tilde{p}_\rv{X}(x)\cdot V_{\rv{Y}|\rv{X}}(y|x) \cdot \log{\frac{V_{\rv{Y}|\rv{X}}(y|x)}{\tilde{p}_\rv{Y}(y)}}} \\
    &\leq \sum_{x}\abs*{p_\rv{X}(x) - \tilde{p}_\rv{X}(x)} \cdot \sum_{y} V_{\rv{Y}|\rv{X}}(y|x)\cdot\log{\frac{1}{V_{\rv{Y}|\rv{X}}(y|x)}} + \sum_{y} \abs*{p_\rv{Y}(y)\log{p_\rv{Y}(y)} - \tilde{p}_\rv{Y}(y)\log{\tilde{p}_\rv{Y}(y)}}\\
    &\leq \size*{\set{X}}\cdot \xi \cdot \log{\size*{\set{Y}}} + \sum_{y} \abs*{p_\rv{Y}(y)\log{p_\rv{Y}(y)} - \tilde{p}_\rv{Y}(y)\log{\tilde{p}_\rv{Y}(y)}}.
    \label{eq:I:continuity:proof:3}
\end{align}
Again, note that 
\begin{align}
\abs*{p_\rv{Y}(y) - \tilde{p}_\rv{Y}(y)}
&= \abs*{\sum_{x} p_\rv{X}(x)\cdot V_{\rv{Y}|\rv{X}}(y|x) - \tilde{p}_\rv{X}(x)\cdot V_{\rv{Y}|\rv{X}}(y|x)} \\
&\leq \sum_{x} \abs*{p_\rv{X}(x) - \tilde{p}_\rv{X}(x)} \cdot V_{\rv{Y}|\rv{X}}(y|x)
\leq \size*{\set{X}} \cdot \xi.
\end{align}
Hence,~\eqref{eq:I:continuity:proof:2} also holds in this case.
Combining~\eqref{eq:I:continuity:proof:2} with~\eqref{eq:I:continuity:proof:3}, we have 
\begin{align}
\abs*{D\infdiv*{p_\rv{X}\cdot V_{\rv{Y}|\rv{X}}}{p_\rv{X}\times p_\rv{Y}} - D\infdiv*{\tilde{p}_\rv{X}\cdot V_{\rv{Y}|\rv{X}}}{\tilde{p}_\rv{X}\times \tilde{p}_\rv{Y}}}
&\leq \size*{\set{X}}\cdot \xi \cdot \log{\size*{\set{Y}}} + \size*{\set{Y}} \cdot \size*{\set{X}} \cdot \xi \cdot \log{\frac{1}{\size*{\set{X}}\cdot \xi}} \\
&= \text{RHS of~\eqref{eq:I:continuity:2:app}}. \qedhere
\end{align}
\end{proof}

\section{Proof of~Eq.~\ref{eq:exp:growth:pertubation}}\label{app:exp:growth:pertubation}
We restate~\eqref{eq:exp:growth:pertubation} as the following lemma.
\begin{lemma*}[Eq.~\eqref{eq:exp:growth:pertubation} restated]
Let $\{a_n\}_{n\in\mathbb{N}}$ be a positive sequence such that $\{\frac{1}{n}\log{a_n}\}_{n\in\mathbb{N}}$ is convergent.
For any $r, \delta>0$, it holds that 
\begin{equation}
\lim_{n\to\infty} \frac{1}{n} \log{a_n} = \lim_{n\to\infty} \frac{1}{n} \log{\left(a_n + \exp\left(-e^{n\cdot\delta} + e^{-n\cdot r}\right)\right)}.\tag{\ref{eq:exp:growth:pertubation}, repeated}
\end{equation}
\end{lemma*}
\begin{proof}
Denote $L\defeq \lim_{n\to\infty} \frac{1}{n} \log{a_n}$.
Let $\eta\in(0,\abs*{L}/2)$ be picked arbitrarily.
There exists $N(\eta)\in\mathbb{N}$ such that $L-\eta \leq \frac{1}{n} \log{a_n} \leq L+\eta$ for all $n\geq N(\eta)$.
Furthermore, there must also exist some $N'(\eta)$ and $N''$ such that $e^{n\cdot\delta}\geq -2n\cdot (L-\eta)$ and $e^{-n\cdot r} \leq \frac{1}{2}\cdot e^{n\cdot \delta}$ for $n\ge N'(\eta)$ and $n\ge N''$, respectively.
Thus, for all $n\geq\max\{N(\eta),N'(\eta), N''\}$, it holds that
\begin{align}
\frac{1}{n} \log{\left(a_n + \exp\left(-e^{n\cdot\delta}+e^{-n\cdot r}\right)\right)}
&\leq \frac{1}{n} \log{\left(a_n + \exp\left(n\cdot (L-\eta)\right)\right)}\\
&\leq \frac{1}{n} \log{\left(2\cdot a_n \right)}
= \frac{1}{n}\log{a_n} + \frac{\log{2}}{n} \\
&\leq L + \eta + \frac{\log{2}}{n}.
\end{align}
Let $n$ tend to infinity, we have
\begin{equation}
\limsup_{n\rightarrow \infty} \frac{1}{n} \log{\left(a_n + \exp\left(-e^{n\cdot\delta}+e^{-n\cdot r}\right)\right)} 
\leq L+\eta.
\end{equation}
Since the above holds for positive $\eta$ that are arbitrarily close to $0$, it must hold that 
\begin{equation}
\limsup_{n\rightarrow \infty} \frac{1}{n} \log{\left(a_n + \exp\left(-e^{n\cdot\delta}+e^{-n\cdot r}\right)\right)} \leq L.
\end{equation}
On the other hand, since $\exp(-e^{n\cdot\delta}+e^{-n\cdot r})>0$, it is immediate that 
\begin{equation}
\liminf_{n\rightarrow \infty} \frac{1}{n} \log{\left(a_n + \exp\left(-e^{n\cdot\delta}+e^{-n\cdot r}\right)\right)} \geq L.
\end{equation}
Finally,~\eqref{eq:exp:growth:pertubation} is the result of combining the above two inequalities.
\end{proof}


\end{document}

%% file: preamble.tex
\usepackage{amsmath, amssymb, amsthm, bbm, mathtools, nccmath}
\usepackage{xparse}
\usepackage{xcolor}
\usepackage{url, hyperref}
\usepackage[caption=false, font=footnotesize]{subfig}
\usepackage[utf8]{inputenc}
\theoremstyle{plain} 
    \newtheorem{theorem}{Theorem}
    \newtheorem{lemma}[theorem]{Lemma}
    \newtheorem*{lemma*}{Lemma}
    
    \newtheorem{proposition}[theorem]{Proposition}
    
\theoremstyle{definition}
    \newtheorem{definition}[theorem]{Definition}
\theoremstyle{remark}
    \newtheorem{remark}[theorem]{Remark}
\renewcommand{\mathbf}{\boldsymbol}

\renewcommand{\varepsilon}{\epsilon}
\renewcommand{\ge}{\geqslant}
\renewcommand{\geq}{\geqslant}
\renewcommand{\le}{\leqslant}
\renewcommand{\leq}{\leqslant}
\renewcommand{\tilde}{\widetilde}
\newcommand*{\set}[1]{\mathcal{#1}} 
\newcommand*{\fnc}[1]{\mathrm{#1}} 
\newcommand*{\rv}[1]{\mathsf{#1}} 
\newcommand*{\rvs}[1]{\mathbf{\mathsf{#1}}} 
\newcommand*{\sys}[1]{\mathsf{#1}} 
\newcommand*{\syss}[1]{\mathbf{\mathsf{#1}}} 
\newcommand\diff{\mathop{}\!\mathrm{d}} 

\newcommand*{\tensor}{\otimes}

\newcommand{\hilbert}{\mathcal{H}}
\newcommand{\uniform}{\mathcal{U}} 
\newcommand{\bid}{\mathbbm{1}}
\newcommand{\id}{\fnc{id}}
\newcommand{\reals}{\mathbb{R}}
\newcommand{\integers}{\mathbb{Z}}
\newcommand{\defeq}{\coloneqq}
\newcommand{\defas}{\eqqcolon}

\newcommand{\mle}{\preccurlyeq}
\DeclarePairedDelimiterX{\ket}[1]{\lvert}{\rangle}{#1}
\DeclarePairedDelimiterX{\bra}[1]{\langle}{\rvert}{#1}
\DeclarePairedDelimiterX{\spr}[2]{\langle}{\rangle}{#1\delimsize\vert#2}
\DeclarePairedDelimiterX{\proj}[1]{\lvert}{\rvert}{#1\delimsize\rangle\!\delimsize\langle#1}
\DeclarePairedDelimiterX{\floor}[1]{\lfloor}{\rfloor}{#1}
\DeclarePairedDelimiterX{\ceil}[1]{\lceil}{\rceil}{#1}
\DeclarePairedDelimiterX{\abs}[1]{\lvert}{\rvert}{#1}
\DeclarePairedDelimiterX{\norm}[1]{\lVert}{\rVert}{#1}
\DeclarePairedDelimiterX{\size}[1]{\lvert}{\rvert}{#1}
\DeclarePairedDelimiterX{\infdiv}[2]{(}{)}{#1\delimsize\Vert#2}
\DeclarePairedDelimiterX{\inner}[2]{\langle}{\rangle}{#1,#2}
\ExplSyntaxOn
\NewDocumentCommand{\multiadjustlimits}{m}
 {
  \group_begin:
  \multiadjustlimits_measure:n { #1 }
  \multiadjustlimits_print:n { #1 }
  \group_end:
 }

\tl_new:N  \l__multiadjustlimits_operator_tl
\tl_new:N  \l__multiadjustlimits_limit_tl

\cs_new_protected:Nn \multiadjustlimits_measure:n
 {
  \clist_map_function:nN { #1 } \__multiadjustlimits_measure:n
 }
\cs_new_protected:Nn \__multiadjustlimits_measure:n
 {
  \__multiadjustlimits_measure:NNn #1
 }
\cs_new_protected:Nn \__multiadjustlimits_measure:NNn
 {
  \tl_put_right:Nn \l__multiadjustlimits_operator_tl { #1 }
  \tl_put_right:Nn \l__multiadjustlimits_limit_tl { #3 }
 }

\cs_new_protected:Nn \multiadjustlimits_print:n
 {
  \clist_map_function:nN { #1 } \__multiadjustlimits_print:n
 }
\cs_new_protected:Nn \__multiadjustlimits_print:n
 {
  \__multiadjustlimits_print:NNn #1
 }
\cs_new_protected:Nn \__multiadjustlimits_print:NNn
 {
  \mathop { \vphantom{\l__multiadjustlimits_operator_tl} \mathopen{} #1 }
  \limits
  \sb{ \vphantom{\cramped{\l__multiadjustlimits_limit_tl}} #3 }
 }

\ExplSyntaxOff

\newcommand\ie{\textit{i.e.}}
\newcommand\eg{\textit{e.g.}}

\newcommand\wrt{w.r.t.~}
\newcommand\aka{a.k.a.~}

\usepackage[maxnames = 99]{biblatex} 